\documentclass[reqno,12pt]{article}
\usepackage{amsfonts, amsmath, amssymb, amsthm}
\usepackage{color}
\usepackage{dsfont}
\usepackage[T1]{fontenc}
\usepackage{authblk}
\usepackage{graphicx,tikz,xcolor}
\usepackage{xcolor,eucal,enumerate,mathrsfs}
\usepackage[normalem]{ulem}
\usepackage{epsfig} %,bbm
\usepackage[latin1]{inputenc}
\usepackage{psfrag}          % replace text in figures
\usepackage[ruled,vlined]{algorithm2e}
\usepackage{cancel}
\usepackage{setspace}

\usepackage{hyperref}

\newcommand{\dd}{\, \mathrm{d}}

\numberwithin{equation}{section}

\usepackage{mathtools}

\mathtoolsset{showonlyrefs}

\newcommand{\eps}{\varepsilon}

\newcommand{\tand}{\quad \text{and} \quad}

\DeclareMathOperator*{\Argmax}{arg\,max}

\usepackage{mystyle}
\newcommand{\bra}[1]{\left \langle #1 \right |}
\newcommand{\ket}[1]{\left | #1 \right \rangle}

\DeclareMathOperator{\argmax}{argmax}

\DeclareMathOperator{\ep}{\varepsilon}

\DeclareMathOperator{\Tr}{\mathrm{Tr}}

\newcommand{\mkP}{\mathfrak{P}}

\newcommand{\cE}{\mathcal{E}}
\newcommand{\cH}{\mathcal{H}}
\newcommand{\cA}{\mathcal{A}}

\newcommand{\fr}{\penalty-20\null\hfill$\blacksquare$} 

\newcommand{\rmH}{\mathrm{H}}

\usepackage{appendix}
\usepackage{enumitem}

\newcommand{\suchthat}{\ensuremath{\ : \ }} % such that inside, for example, the sets definition              

\newcommand{\restr}[1]{\lower3pt\hbox{$|_{#1}$}}

\newcommand{\argmmax}{\mathop{\mathrm{argmax}}}

\newcommand{\DM}{\mathfrak{P}}

\newcommand{\prim}{\mathfrak{F}^{\eps}}
    \newcommand{\primtaun}{\mathfrak{F}^{\eps,\tau,n}}
\newcommand{\primtau}{\mathfrak{F}^{\ep,\tau}}

\newcommand{\primal}{\tF^{\eps}}
    \newcommand{\primalun}{\tF^{\eps,\tau,n}}
\newcommand{\primaltau}{\tF^{\ep,\tau}}

\newcommand{\dual}{\mathfrak{D}^{\eps}}
    \newcommand{\dualun}{\mathfrak{D}^{\eps,\tau,n}}    
\newcommand{\dualu}{\mathfrak{D}^{\eps,\tau}}
\newcommand{\dualf}{{\rm D}^\eps}
    \newcommand{\dualufn}{\text{D}^{\eps,\tau,n}}  
\newcommand{\dualfu}{{\rm D}^{\eps,\tau}}

\theoremstyle{plain}
\newtheorem{theo}{Theorem}[section]
\newtheorem{proposition}[theo]{Proposition}
\theoremstyle{definition}
\newtheorem{defin}[theo]{Definition}
\newtheorem{rem}[theo]{Remark}
\newtheorem{lem}[theo]{Lemma}
\newtheorem{coro}[theo]{Corollary}

\newcommand{\tP}{\operatorname{P}}

\newcommand{\tF}{\operatorname{F}}

\newcommand{\tH}{\operatorname{H}}

\DeclareFontFamily{U}{mathx}{\hyphenchar\font45}
\DeclareFontShape{U}{mathx}{m}{n}{
	<5> <6> <7> <8> <9> <10>
	<10.95> <12> <14.4> <17.28> <20.74> <24.88>
	mathx10
}{}
\DeclareSymbolFont{mathx}{U}{mathx}{m}{n}
\DeclareMathSymbol{\bigtimes}{1}{mathx}{"91}

\usepackage{scalerel}

\usepackage{color}
\usepackage{pstricks}
\usepackage{ifthen}
\newrgbcolor{kascolor}{.7 .1 .1}
\newrgbcolor{olicolor}{.1 .7 .1}
\newrgbcolor{augucolor}{.1 .1 .7}
\newrgbcolor{lorcolor}{.5 .5 .5}
\usepackage{csquotes}

\newboolean{DEBUG}
\setboolean{DEBUG}{true}

\ifthenelse {\boolean{DEBUG}}
{\newcommand{\nat}[1]{{\kascolor{[Nataliia: #1]}}}}

\ifthenelse {\boolean{DEBUG}}
{}

\ifthenelse {\boolean{DEBUG}}
{}

\ifthenelse {\boolean{DEBUG}}
{}

\ifthenelse {\boolean{DEBUG}}
{}

\title{Quantum optimal transport with convex regularization}

\author[1]{Emanuele Caputo}
\author[2,3,4]{Augusto Gerolin}
\author[2]{Nataliia Monina}
\author[5]{Lorenzo Portinale}
\affil[1]{Zeeman Building, University of Warwick}
\affil[2]{Department of Mathematics and Statistics, University of Ottawa}
\affil[3]{Department of Chemistry and Biomolecular Sciences, University of Ottawa}
\affil[4]{Nexus for Quantum Technologies, University of Ottawa}
\affil[5]{Institut f\"ur Angewandte Mathematik, Universit\"at Bonn}

\begin{document}
\maketitle
\begin{abstract}

\noindent
In this manuscript, we study the (finite-dimensional) static formulation of quantum optimal transport problems with general convex regularization and its unbalanced relaxation. In both cases, we show a duality result, characterizations of minimizers (for the primal) and maximizers (for the dual).
% , and limit behaviors for the regularization parameter going to zero \lore{Well, not really for now, but we may do that as well -- probably we should}. 
An important tool we define is a non-commutative version of the classical $(c,\psi, \ep)$-transforms associated with a general convex regularization, which we employ to prove the convergence of the associated Sinkhorn iterations. Finally, we show the convergence of the unbalanced transport problems towards the constrained one, as well as the convergence of transforms, as the marginal penalization parameters go to $+\infty$.
\end{abstract}

% ABSTRACT QOTTA
% We study both the balanced and the unbalanced quantum optimal transport problem in their static, finite dimensional formulations. We discuss duality results, characterisations for optimisers, and limit behaviors for the marginal trade-off parameter going to +∞, i.e. convergence results from unbalanced to balanced. An important tool we define and analyse is a non-commutative version of the classical (c, ε)-transforms associated with a general convex regularization.

\tableofcontents

\section{Introduction}

This work provides the first complete study of convex regularized quantum optimal transport problems and their extension to the unbalanced case.  More precisely, given two finite-dimensional Hilbert spaces $\cH_1$ and $\cH_2$,  our goal is to prove the existence and to characterize the solutions of the convex regularized quantum optimal transport problem
\begin{equation}\label{eq:general_pb_intro}
     \prim = \min 
    \left\{ 
        \Tr\left[C \Gamma\right] + \eps \Tr \left[ \varphi(\Gamma)\right]
            \suchthat 
        \Gamma \mapsto (\rho,\sigma) 
    \right\}, 
\end{equation}
where the cost $C \in \rmH(\cH_1\otimes\cH_2)$ is a Hermitian matrix over the product space $\cH_1\otimes \cH_2$, and $\varphi:[0,+\infty) \to \mathbb{R}$ is a proper, convex, superlinear at infinity and bounded from below function. The notation $\Gamma\mapsto(\rho,\sigma)$ means that $\Gamma$ is a density matrix over $\cH_1\otimes\cH_2$ having partial traces equal to density matrices $\rho$ over $\cH_1$  and  $\sigma$ over $\cH_2$.

The unbalanced (convex-regularized) quantum optimal transport problem is given by
\begin{equation}\label{eq:general_unbalancedpb_intro}
    \mathfrak{F}^{\ep,\tau_1,\tau_2} = \min 
    \left\{ 
        \Tr\left[C \Gamma\right] + \eps \Tr \left[\varphi(\Gamma) \right]
           + \tau_1 \cE({\rm P}_1(\Gamma)| \rho)+ \tau_2 \cE({\rm P}_2(\Gamma)| \sigma) \, : \, \Gamma\geq 0
    \right\}, 
\end{equation}

where $\cE(\Gamma|\rho)= 
    \Tr
        \left[
            \Gamma 
            \big(
                \log \Gamma - \log \rho - {\rm Id}_{\cH}
            \big)
                + \rho 
        \right]$ 
denotes the non-commutative relative entropy functional, $\rm{P}_1$ and $\rm{P}_2$ are, respectively, the partial traces in $\cH_2$ and $\cH_1$. The notation $\Gamma\geq 0$ means that $\Gamma$ is a semi-definite positive operator. For the precise definitions, see equation \eqref{eq:def_primal_problem_balanced} and Section~\ref{sec:setting_mr}.

The main difference between the problems \eqref{eq:general_pb_intro} and \eqref{eq:general_unbalancedpb_intro} is that instead of the marginal constraint $\Gamma\mapsto(\rho,\sigma)$ in \eqref{eq:general_pb_intro}, the unbalanced quantum optimal transport introduces a soft-penalization for the partial traces of $\Gamma$ to be close (in the relative entropy functional) to the density matrices $\rho$ and $\sigma$, weighted by some trade-off parameters $\tau_1$,$\tau_2 \in (0,+\infty)$.

The main contributions of this paper include
\begin{enumerate}
    \item A general duality result, both for the balanced \eqref{eq:general_pb_intro} and unbalanced case \eqref{eq:general_unbalancedpb_intro}, including the case when $\varphi$ nor its Legendre's transform are nonsmooth (Theorem~\ref{thm:duality_general_intro}).
    \item The study and characterization of the optimizers, both in \eqref{eq:general_pb_intro} and \eqref{eq:general_unbalancedpb_intro}, as well as their corresponding dual formulations (Theorem~\ref{thm:main_duality_balanced} and \ref{thm:main_duality_unbalanced}). 
    \item The definition and the proof of the convergence of Sinkhorn-type iterations for quantum optimal transport with general convex regularization $\varphi$ (Theorem~\ref{thm:main_Sinkhorn}).
    \item The convergence, as the soft-penalization parameters $\tau_i \to 
    +\infty$, of the unbalanced transport problem \eqref{eq:general_unbalancedpb_intro} to the corresponding balanced formulation \eqref{eq:general_pb_intro}. More precisely, we show that when $\tau_i\to+\infty$ o(i) the energies converge, i.e.  $\mathfrak{F}^{\ep,\tau_1,\tau_2} \to\prim$; (ii) Any sequence of minimizers $(\Gamma^{\varepsilon,\tau_1,\tau_2})_{\tau_1,\tau_2>0}$ of the unbalanced quantum optimal transport \eqref{eq:general_unbalancedpb_intro} converges to a minimizer $\Gamma^{\varepsilon}$ of \eqref{eq:general_unbalancedpb_intro}; and finally (iii) any sequence of maximizers of the dual formulation of \eqref{eq:general_unbalancedpb_intro} converges to a maximizer in the dual formulation of \ref{eq:general_unbalancedpb_intro} 
 (Theorem~\ref{thm:main_convergence}).
\end{enumerate}
Moreover, if we assume that $\varphi$ is strictly convex and $C^1$, we can show that for both primal problems defined in equations \eqref{eq:general_pb_intro} and \eqref{eq:general_unbalancedpb_intro} admit a unique minimizer, which can be written in the explicit form
\begin{align}
    \Gamma
        =
    \psi' \left[ \frac{U \oplus V - C}{\varepsilon} \right]
        , 
\end{align}
where $\psi$ is the Legendre transform of $\varphi$, and $(U,V)$ are maximizers of their associated dual problems (Theorem \ref{thm:main_duality_balanced} and Theorem \ref{thm:characterization_maximizers_unbalanced}). 

\subsubsection*{Quantum Optimal Transport}

Quantum optimal transport extends the classical theory of optimal transport for probability measures to quantum states (e.g., density matrix or operators). Mathematically, challenges arise from the intrinsic non-commutative nature of quantum mechanics, which contrasts with the commutative structures underlying classical probability theory.

The pursuit of a quantum version of the Monge--Kantorovich optimal transport theory dates back roughly thirty years, with pioneering works by Connes \cite{Con89}, Biane, and Voiculescu \cite{BiaVoi01}, see also \cite{DAnMar10}. Their approach is primarily centered on adapting the static dual formulation of optimal transport.

More recently, several efforts have been made to extend the optimal transport theory for quantum states, starting from either the static (e.g., Kantorovich relaxation, dual) or dynamical formulation (known after Benamou--Brenier). While these formulations are equivalent in the optimal transport theory for probability measures, they can lead to distinct and not necessarily equivalent theories in the quantum setting.

In the static framework, quantum optimal transport between quantum states has been first introduced by Caglioti, Golse, and Paul \cite{CalGolPau18} and De Palma and Trevisan \cite{DPaMaTre2021,DPaTre19}. In \cite{CalGolPau18}, the quantum optimal transport arose in the context of
the study of the semiclassical limit of quantum mechanics. In this setting, the problems correspond to \eqref{eq:general_pb_intro} with $\varepsilon=0$, the joint states $\Gamma$ are defined in $\cH_1\otimes \cH_2$ and the cost operator is given in terms of quadratures $C=\sum_{i=1}^{2m}\left(R_i\otimes\mathbb{I}_{\mathcal{H}_2} - \mathbb{I}_{\mathcal{H}_1}\otimes R_i\right)^2.$
In the alternative proposal in \cite{DPaTre19}, the plans instead are states on $\mathcal{H}\otimes \mathcal{H}^*$, i.e., one acts with a partial transpose operation:
%The cost is then
%$$ C = \sum_{i=1}^{2m}\left(R_i\otimes\mathbb{I}_{\mathcal{H}^*} - \mathbb{I}_{\mathcal{H}}\otimes R_i^T\right)^2,$$
%and
the variational problem reads similarly to the previous one, but its properties are rather different, and it enjoys operational interpretations thanks to a correspondence between plans and quantum channels. Also, in \cite{DPaMaTre2021}, the authors introduced a quantum optimal transport problem by generalizing the dual formulation of the so-called $\mathbb{W}_1$ distance specifically for qubits.

A dynamical formulation of quantum optimal transport has been introduced by Carlen and Maas \cite{CarMas14,CarMas20}, which the authors used to study long-behavior and fine properties of the Fermionic Fokker--Planck equations. In \cite{Wirth:2021}, Wirth introduced a dynamical formulation for von Neumann entropy-regularized quantum optimal transport within the dynamical framework. In the same year, Monsaingeon and Vorotnikov \cite{MonVor20} introduced the Fisher--Rao space of matrix-valued non-commutative probability measures, and the related Hellinger space.

The study of quantum composite systems at positive temperature -- equation \eqref{eq:general_pb_intro} with von Neumann entropy $\varphi(\Gamma)=\Tr\left[\Gamma\big(\log \Gamma - {\rm Id}_{\cH_1\otimes\cH_2}\big)\right]$ -- has been initiated in \cite{FelGerPor23}, including the cases when the density matrices are fermionic, important for Electronic Structure Theory \cite{AyeGolLev06,FriGerGorChapter,Gil75,Levy,Val80}. The von Neumann entropy regularization is fundamental for implementing quantum optimal transport functionals into efficient computational algorithms. 

Unbalanced quantum optimal transport was studied in \cite{GerMon2023}, motivated 
by optimal transport theory for non-negative measures introduced, independently, by Liero, Mielke, and Savar\'e \cite{LiMiSa}, Chizat, Peyr{\'e}, Schmitzer, and Vialard \cite{ChiPeySchVia18}, and finally by Kondratyev, Monsaingeon, and Vorotnikov \cite{KonMonVor16,KonMonVor16-arX}. In \cite{GerMon2023},
the authors considered a particular problem \eqref{eq:general_pb_intro} when $\varphi(\Gamma)$ is the von Neumann entropy regularization, proved weak duality, $\Gamma$-convergence results when either the soft-penalization parameter in \eqref{eq:general_unbalancedpb_intro} 
 $\tau_1=\tau_2\to+\infty$ goes to infinity and the von Neumann entropy-regularization parameter $\varepsilon\to 0^+$ goes to zero.

This work addresses general convex entropy regularization problems within both the quantum optimal transport \eqref{eq:general_pb_intro}
and unbalanced quantum optimal transport \eqref{eq:general_unbalancedpb_intro} frameworks. In the context of optimal transport for probability measures, convex regularization has been studied extensively (e.g., \cite{DMaGer19, DMaGer20,LorMah20Cont,LorMah19, LorManMey21, Muzellec2016TsallisRO, TerGon2022}). This list of examples includes quadratic regularization as well as regularization with the Tsallis entropy. The study of these cases is primarily driven by the goal of designing computational algorithms leveraging the sparse structure of optimal transport problems.

A recent survey on the topic has been recently published in \cite{DPaTre2023,PorSurvey23}.

\noindent
\textbf{Organization of the paper:} Section \ref{sec:setting_mr} presents the setting and the statement of the main result of this manuscript. Section \ref{sec:balanced_noncommutative} focuses on the proofs for the convex regularized quantum optimal transport problem, the definition of $(C,\psi,\eps)$-transform and its properties. In the final part, we define the Sinkhorn iterations and study their convergence. In Section \ref{sec:unbalanced_noncommutative}, we prove strong duality in the unbalanced case, we define the $(C,\psi,\eps,\tau)$-transform and study its properties. Section \ref{sec:extension_to_nonsmooth} extends the results of previous sections by removing the assumption that $\psi \in C^1$. Section \ref{sec:convergence_resuts} contains the $\Gamma$-convergence results of the unbalanced problem to the balanced one and the convergence of transforms as the size parameters $\tau$ converge to $+\infty$.

\section{Setting and main results}
\label{sec:setting_mr}
%
% \nat{
% $\rmH(\cH)$-- hermitian ;
% $\rmH_\geq(\cH)$-- SDP ;
% $\mkP(\cH)$-- DP ;
% $\mkP(\cH)$ -- density matrices  ;
% }
Given a Hilbert space $\cH$, we denote by $\rmH(\cH)$ the vector space of Hermitian operators over $\cH$.
Let us denote by $\rmH_\geq(\cH)$ the set of positive semi-definite Hermitian and by $\mkP(\cH)$ the set of density matrices, i.e.\ elements of $\rmH_\geq(\cH)$ with trace equal to $1$. We also write 
$\rmH_>(\cH)$ to denote the elements $A$ of $\rmH_\geq(\cH)$ such that ${\rm ker}(A)=\{ 0\}$.

Given $\rho \in \rmH(\cH_1)$ and $\Gamma \in \rmH(\cH_1 \otimes \cH_2)$, we say that $\rho$ is the first partial trace if
\begin{align}
\label{eq:def_partial_traces}
    \Tr \left[ \Gamma (U \otimes {\rm Id}_{\cH_2}) \right]= \Tr \left[ \rho U \right]\qquad \text{for every }U \in \rmH(\cH_1) .
\end{align}
In such a case, we write ${\rm P}_1 \Gamma = \Tr_{\cH_2} \Gamma =\rho$. Similarly, given $\sigma \in \rmH(\cH_2)$ and $\Gamma$ as before, we denote by ${\rm P}_2 \Gamma = \Tr_{\cH_1} \Gamma =\sigma$ the second partial trace defined in an analogous way. If ${\rm P}_1 \Gamma=\rho$ and ${\rm P}_2 \Gamma=\sigma$, we write $\Gamma \mapsto (\rho, \sigma)$. 

For the sake of convenience, we introduce the following notation: for $ \in \{1, \dots, d \}$, we use $\lambda_i(A)$ to denote the $i-$th smallest element of the spectrum of $A$. In particular, $\lambda_1(A)$ is the smallest eigenvalue of $A$. 
For every continuous function $\varphi \colon \mathbb{R} \to \mathbb{R}$, we define \textit{the lifting} of $\varphi$ to the space of Hermitian matrices the operator\footnote{With a slight abuse of notation, we denote the lifting of $\varphi$ also with the same letter $\varphi$.} $\varphi \colon \rmH(\cH) \to \rmH(\cH)$ given by $\varphi(A):=\sum_{i=1}^d \varphi(\lambda_i(A)) |\xi_i\rangle \langle \xi_i|$, where we used the spectral decomposition $A = \sum_{i=1}^d \lambda_i(A) \ket{\xi_i}\bra{\xi_i}$. When clear from the context, we may omit the dependence on $A$ and write $\lambda_i$ instead of $\lambda_i(A)$. 
For $U \in \rmH(\cH_1)$ and $V \in \rmH(\cH_2)$, we define $U \oplus V := {\rm Id}_{\cH_1} \otimes V + U \otimes {\rm Id}_{\cH_2}$.
%
%

% \nat{Something is wrong with this sentence -->} 
The unbalanced optimal transport is obtained by replacing the constraint with two penalization terms tuned by parameters $\tau_1$, $\tau_2>0$. The choice of penalization we use in this paper (as typically adopted in the commutative setting as well) is given by the relative entropy with respect to the marginals. Given a finite-dimensional Hilbert space $\cH$ and an operator $\eta \in \rmH_\geq(\cH)$, the relative entropy functional $\cE(\cdot | \eta): \rmH_\geq(\cH) \to [0, + \infty]$ is the functional defined as
\begin{align}
\label{eq:def_relative_entropy}
    \alpha \mapsto \cE(\alpha | \eta)
        := 
    \begin{cases}
        \Tr 
        \left[
            \alpha 
            \big(
                \log \alpha - \log \eta - {\rm Id}_{\cH}
            \big)
                + \eta 
        \right]
            & \text{if } \ker \eta \subset \ker \alpha
                , 
    \\
        + \infty 
            & \text{otherwise} 
                .
    \end{cases}
\end{align}
Here, we are using the convention that $\alpha \big( \log \alpha - \log \eta) = 0$ on $\ker \eta \subset \ker \alpha$. The nonnegativity of the relative entropy is a direct consequence of Klein's inequality \cite[Theorem~2.11]{carlen2010trace} applied to the convex function $t \mapsto t \log t$. Finally, it is a (strictly) convex function on its domain, and its Legendre transform is given by 
\begin{align}
\label{eq:def_Legendre_relative_entropy}
    \cE^*(A|\eta)
        :=
    \Tr
    \left[ 
        e^{A+\log \eta}  - \eta 
    \right]
    % { \color{black} \neq 
    % \left[ 
    %     \big( 
    %         e^A - {\rm Id}_{\cH}
    %     \big) 
    %         \eta 
    % \right]
    % }
    % This one below is wrong, watch out!
    % \left[ 
    %     \big( 
    %         e^A - {\rm Id}_{\cH}
    %     \big) 
    %         \eta 
    % \right]
        \, , \qquad A \in \rmH(\cH) 
            .
\end{align}
This follows from the equality 
\begin{align}
     \Tr
    \left[ 
        e^{A+\log \eta}  - \eta 
    \right] 
        =
     g(A+\log \eta) - \Tr[\eta] 
        \, , \quad \text{where} \quad g(B) := \Tr[\exp(B)]
            \, , 
\end{align}
together with the fact that, for $\alpha \in \rmH(\cH)$,
\begin{align}
    g^*(\alpha) 
        = 
    \sup_{A \in \rmH(\cH)} 
    \left\{ 
        \langle \alpha, A \rangle - g(A)
    \right\} 
        = 
        \begin{cases}
            \Tr 
        \left[
            \alpha 
            \big(
                \log \alpha - {\rm Id}_{\cH}
            \big)
        \right], & \alpha\in \rmH_\geq(\cH),\\
        +\infty, & \text{otherwise}
            .
        \end{cases}
\end{align}
By convexity we know that $g^{**} = g$,  hence we get  \eqref{eq:def_Legendre_relative_entropy}.
Note that in general  $ \cE^*(A|\eta) \neq \Tr 
    \left[ 
        \big( 
            e^A - {\rm Id}_{\cH}
        \big) 
            \eta 
    \right]
    $, in particular when $A$ and $\eta$ do not commute.
% \lore{We could put this in appendix, but the reason is
% \begin{align}
%     \cE^*(A|\eta)
%         =
%     g(A+\log \eta) - \Tr[\eta] 
%         \, , \quad \text{where} \quad g(B) := \Tr[\exp(A)]
% \end{align}
% }
% Maybe write it in appendix.
% This can be computed directly from the definition of Legendre transform. 

We shall fix  $\varphi \colon [0,+\infty) \to \R$ a convex, superlinear at infinity, and bounded from below function, namely 
\begin{align}
    \label{eq:assumptions_varphi_intro}
    \varphi \colon [0,+\infty) \to \R
        \, , \quad 
    \text{convex} 
        \, , \quad 
    \lim_{t \to +\infty}
        \frac{\varphi(t)}{t} =+\infty
        \, ,  \tand 
    \inf \varphi \geq l > - \infty 
        , \qquad 
\end{align}
for some $l \in \R$. In particular, we assume that $\varphi(0)\in \R$. 
% \nat{Do we need to put already here that $\varphi$ is defined on the whole line and is $+\infty$ on negative values? In particular, we clearly need this in the statement of the Duality theorem because we say $\psi = \phi^*$ by assumption.}
% 
% 
The following operators are given and fixed throughout the whole paper: 
\begin{enumerate}
    \item A cost operator $C \in \rmH(\cH_1 \otimes \cH_2)$. 
    \item Two `marginal' operators $\rho \in \rmH_>(\cH_1)$ and $\sigma \in \rmH_>(\cH_2)$ (in particular, with trivial kernel). Sometimes we may assume them to be density matrices as well.  
\end{enumerate}

The assumption on the kernel of the marginal is classical and not too restrictive. Similar results can be obtained by decomposing the space with respect to the kernels and their orthogonal spaces, in the very same spirit as in \cite[Remark~3.9]{FelGerPor23}. For the sake of simplicity, we won't discuss further details in this work. 

We also fix $\eps >0$ a regularization parameter. 
The main objects of study of our work are primal and dual functionals. 
\begin{defin}[Primal problems]
We define a functional $\primal \colon \rmH_\geq(\cH_1 \otimes \cH_2) \to \mathbb{R}$ as
\begin{align}
\label{eq:def_primal_balanced}
    \primal(\Gamma): = \Tr\left[C \Gamma\right] + \eps \Tr \left[ \varphi\left( \Gamma \right) \right]
        , \quad \forall \Gamma \in \rmH_\geq(\cH_1 \otimes \cH_2) .
\end{align}
For every $\tau = (\tau_1, \tau_2) \in (0,+\infty)^2$, we define  $\primaltau: \rmH_\geq(\cH_1 \otimes \cH_2) \to \R$ as the functional 
\begin{align}
\label{eq:def_primal_unbalanced}
    \primaltau(\Gamma) :=  \primal(\Gamma)+ \tau_1 \cE({\rm P}_1(\Gamma)| \rho)+ \tau_2 \cE({\rm P}_2(\Gamma)| \sigma) 
        ,
\end{align}
We also define the \textit{unbalanced} primal problem with parameters $\tau_1$, $\tau_2>0$ as
\begin{align}
\label{eq:def_primal_problem_unbalanced}
    \primtau:= \min \left\{ \primaltau (\Gamma) 
    \suchthat \Gamma \in \rmH_\geq(\cH_1 \otimes \cH_2) \right\}.
\end{align}
Moreover, if $\rho \in \mkP(\cH_1)$ and $\sigma \in \mkP(\cH_2)$ we define the primal problem as 
\begin{align}
\label{eq:def_primal_problem_balanced}
    \prim:= \min \left\{ \primal(\Gamma) \suchthat  \Gamma \mapsto (\rho, \sigma)  \right\}
    ,
\end{align}
which (formally) corresponds to the case $\tau_1=\tau_2 = +\infty$.
\end{defin}

To each primal problem, we associate the respective dual problem. Throughout the whole paper, we work with a function $\psi \in C(\R)$ which is convex, superlinear at infinity, and bounded from below, i.e. 
\begin{align}
\label{eq:assumptions_psi_intro}
    \psi\in C(\R) 
        \, , \quad \text{convex} \, , \quad 
    \lim_{t \to +\infty}
        \frac{\psi(t)}{t}
    = + \infty \, , 
        \tand 
    m:= \inf \psi > - \infty 
        .
\end{align}
When dealing with duality results for the quantum optimal transportation problems, $\psi$ will typically be the Legendre transform of a $\varphi$ satisfying \eqref{eq:assumptions_varphi_intro}, namely of the form
\begin{align}
    \psi(t) 
        =
    \sup_{x \in [0,+\infty)}
        \left\{
            tx - \varphi(x)
        \right\}
        =
    \varphi^*(t)
        \, , \quad \forall t \in \R
            ,
\end{align}
where when talking about Legendre transform we may implicitly extend the definition of $\varphi$ on the full real line by setting $\varphi(x) \equiv +\infty$, for every $x < 0$. The validity of \eqref{eq:assumptions_psi_intro} readily follows in this case from the properties \eqref{eq:assumptions_varphi_intro} of $\varphi$. 
A direct computation for the relation between the primal and the dual can be seen in Appendix \ref{appendix:primal_dual}.

\begin{defin}[Dual problems]
Fix $\psi \in C(\R)$ satisfying \eqref{eq:assumptions_psi_intro}. We define the dual functional $\dualf  \colon \rmH(\cH_1) \times \rmH(\cH_2) \to \R$ as
\begin{align}\label{eq:def_dual_balanced}
    % \dual(U,V)
    \dualf (U,V)
        :=
    \Tr
    \left[ U \rho\right] 
        + 
    \Tr\left[ V \sigma\right]
        -
    \eps \Tr \left[ \psi \Big(\frac{U\oplus V - C}{\eps} \Big) \right]
        .
\end{align}
% \nat{We need to edit where $\rho$ and $\sigma$ belong to, they must be definite positive} 
%[Lore]: keep an eye on when/where we use it/
For every $\tau_1$, $\tau_2>0$, we also define $\dualfu \colon \rmH(\cH_1) \times \rmH(\cH_2) \to \R$ as the map
\begin{align}
\label{eq:def_dual_unbalanced}
    % \dualfu (U,V)
        % :=
    (U,V) \mapsto 
     -\tau_1
     \Tr
     \left[ 
        e^{-\frac{U}{\tau_1}+\log \rho}- \rho
    \right] 
        -
    \tau_2
    \Tr
    \left[ 
       e^{-\frac{V}{\tau_2}+\log \sigma}- \sigma
    \right] 
        -\eps 
    \Tr
    \left[ 
        \psi\Big(\frac{U\oplus V-C}{\eps}\Big)
    \right] 
        .
    \quad 
\end{align}
We define the dual problem and the unbalanced dual problem as 
\begin{gather}
\label{eq:def_dual_problems}
    \dual :=
        \sup 
    \Big\{
        \dualf (U,V)
    \suchthat 
         (U,V) \in \rmH(\cH_1) \times \rmH(\cH_2)
    \Big\}
        ,
\\
    \dualu := 
        \sup
    \Big\{ 
        \dualfu (U,V) 
    \suchthat 
        (U,V) \in \rmH(\cH_1) \times \rmH(\cH_2)
    \Big\}
        .  
\end{gather}
\end{defin}

% LORE: I have decided it was a useless comment lol
% \begin{rem}[Equivalent formulation of the unbalanced problem]
% By performing the change of variable given by 
%     \begin{align}
%         \tilde U = U - \tau_1 \log \rho 
%             \quad \tand \quad  
%          \tilde V = V - \tau_2 \log \sigma 
%             , 
%     \end{align}
%     we immediately see that the dual problem $\dualu$ associated to a cost operator $C$ coincides with $\mathfrak{D}_{{\rm Id}, {\rm Id}}^{\eps,\tau}$ (up to an additional term which only depends on $\rho$ and $\sigma$) with cost operator given by $\tilde C:= C + (\tau_1 \log \rho) \oplus (\tau_2 \log \sigma)$. \lore{We can also erase this, or explain more} \fr
% \end{rem}
% \lore{
% state the main result with the assumption that the kernel of the marginals are trivial, then discuss/put a reference to the old paper on how to handle the case with nontrivial kernel.
% % 
% }

The main results of our work can be divided into four parts: duality results and analysis of primal/dual problem for the convex regularized quantum optimal transport problem under suitable regularity assumptions on $\psi$, the corresponding results for the unbalanced problem, a generalization of the duality for possibly nonsmooth $\psi$, and finally convergence results for the unbalanced problem towards the balance one was $\tau \to +\infty$. 

\subsubsection*{Quantum optimal transport}
The first contribution is a duality result for the convex regularized quantum optimal transport problem, including existence, (suitable) uniqueness, and characterization of the optimizers, both in the primal and dual problems.

\begin{theo}[Duality and optimizers]
\label{thm:main_duality_balanced}
    Let $\varphi:[0,+\infty) \to [0,+\infty)$ satisfy \eqref{eq:assumptions_varphi_intro}, and assume that $\psi=\varphi^*$ is strictly convex and $C^1$. Assume additionally that $\rho \in \DM(\cH_1)$ and $\sigma \in \DM(\cH_2)$. Then the following statements hold: 
    \begin{itemize}
        \item[1)] (Duality) The dual and primal problem coincide, namely $\dual  = \prim$.
        \item[2)] (Existence of maximizers) There exists $(U^*,V^*)$ which maximizes $\emph{D}^\eps $, i.e. $\dual  = \emph{D}^\eps (U^*,V^*)$. Moreover, any other maximizer $(\tilde U, \tilde V)$ satisfies $(\tilde U - U^*, \tilde V - V^*) = (\lambda {\rm Id}_{\cH_1} , - \lambda {\rm Id}_{\cH_2} )$, for some $\lambda \in \R$. 
        \item[3)] (Existence of minimizers) There exists a unique maximizer $\Gamma^* \in \DM(\cH_1 \otimes \cH_2)$ for the primal problem, i.e. $\prim = \primal(\Gamma^*)$, and it is given by 
        
        \begin{align}
        \label{eq:main_formula_minimizer}
            \Gamma^*
                =
            \psi' \left[ \frac{U^* \oplus V^* - C}{\varepsilon} \right]
                ,
        \end{align}
        where $(U^*,V^*)$ are (any) maximizers for the dual problem. 
    \end{itemize}
\end{theo}

Our second contribution concerns the definition of a Sinkhorn-type iterations and their convergence guarantee. To this end, a crucial tool to consider is the so-called $(C,\psi,\eps)$-transform (associated with $\psi$). We assume once again that  $\psi=\varphi^*$ is strictly convex and $C^1$. We define $\mathscr{F}_2^{(C,\psi,\eps)} \colon \rmH(\cH_1)\to \rmH(\cH_2)$ as
\begin{equation}\label{eq:U_transform_bal_intro}
    \mathscr{F}_2^{(C,\psi,\eps)}(U) 
        =
    \argmax \left\{ \Tr[ V \sigma] - \eps \Tr \left[ \psi \Big(\frac{U\oplus V - C}{\eps} \Big) \right]
        \suchthat 
    V \in \rmH(\cH_2) \right\}.
\end{equation}
Analogously, we define the corresponding $\mathscr{F}_1^{(C,\psi,\eps)} \colon \rmH(\cH_2)\to \rmH(\cH_1)$ as
\begin{equation}
\label{eq:V_transform_bal_intro}
    \mathscr{F}_1^{(C,\psi,\eps)}(V) 
        = 
    \argmax 
    \left\{ 
        \Tr[ U \rho] - \eps \Tr \left[ \psi \Big(\frac{U\oplus V - C}{\eps} \Big) \right] 
            \suchthat 
        U \in \rmH(\cH_1) 
    \right\}.
\end{equation}
For the well-posedness of the $(C,\psi,\eps)$-transforms, see Section~\ref{sec:Cepspsi_transform}.
We use this notion to define the following algorithm. Recall that $\lambda_1(A)$ denotes the first eigenvalue of an operator $A$.

\underline{\emph{Step 0}}: fix any initial $(U^0,V^0)\in \rmH(\cH_1)\times\rmH(\cH_2)$. 

\underline{\emph{Step 2n-1}}: for $n \in \N$, for given $U^{n-1} \in \rmH(\cH_1)$, we define $V^n$ as 
\begin{align}
    V^n := \mathscr{F}_2^{(C,\psi,\eps)}(U^{n-1}) - \lambda_1(\mathscr{F}_2^{(C,\psi,\eps)}(U^{n-1})){\rm Id}_{\cH_2}
        .
\end{align}

\underline{\emph{Step 2n}}:  for given $V^n \in \rmH(\cH_2)$, we define $U^n$ as 
\begin{align}
     U^n := \mathscr{F}_1^{(C,\psi,\eps)}(\mathscr{F}_2^{(C,\psi,\eps)}(U^{n-1})) + \lambda_1(\mathscr{F}_2^{(C,\psi,\eps)}(U^{n-1})){\rm Id}_{\cH_1}.
\end{align}

Then we have the following convergence result.

\begin{theo}[Convergence of the Sinkhorn iterations]
% {\color{red} Ema: I would like to add some assumptions on $\psi$ here like previous theorem.}
\label{thm:main_Sinkhorn}
Under the same assumptions of Theorem~\ref{thm:main_duality_balanced},
    there exists a limit point $(U^*,V^*)$ of  $\{(U^n,V^n)\}_{n\geq 0}$, which is a maximizer of $\emph{D}^\eps $.
    % \item[(b)] There exists a sequence $\{\mu^n\}_{n\geq 1}\subset \R$ and $(U^*,V^*)\in \rmH(\cH_1)\times\rmH(\cH_2)$ such that 
    % \[
    %     (U^n+\mu^n,V^n-\mu^n) \xrightarrow[]{n\to\infty} (U^*,V^*),
    % \]
    % and $(U^*, V^*)$ is a maximizer of $\dualf $.
Moreover, we have that 
    \begin{align}
        \Gamma^n = \psi'
        \Big(
            \frac{U^n\oplus V^n - C}{\ep}
        \Big)
            \in 
                \mkP(\cH_1 \otimes \cH_2)
    \end{align}
    and we have that $\displaystyle \Gamma^n \to \Gamma^* \in \mkP(\cH_1 \otimes \cH_2)$ as $n \to \infty$, where 
    \begin{align}
        \Gamma^* 
            = 
        \psi'
        \Big( 
            \frac{U^* \oplus V^* - C}{\ep}
        \Big)
            \mapsto 
        (\rho,\sigma) 
    \end{align}
    is the unique minimizer for the primal problem $\prim$.
\end{theo}

\subsubsection*{Unbalanced quantum optimal transport}
Taking advantage of the techniques and results provided in the previous section, we next tackle the same questions in the unbalanced setting. In particular, we have the validity of the following duality result. 

\begin{theo}[Duality and optimizers in the unbalanced case]
\label{thm:main_duality_unbalanced}
     Let $\varphi:[0,+\infty) \to [0,+\infty]$ satisfy \eqref{eq:assumptions_varphi_intro}, and assume that $\psi=\varphi^*$ is strictly convex and $C^1$. Let $\tau = (\tau_1,\tau_2) \in \R^2$ be trade-off parameters. Then the following holds:
    \begin{itemize}
        \item[1)] (Duality) The dual and primal problem coincide, namely $\dualu  = \primtau$.
        \item[2)] (Existence of maximizers) There exists a unique $(U^{*,\tau},V^{*,\tau})$ which maximizes $\emph{D}^{\eps,\tau} $, i.e. $\dualu  = \emph{D}^{\eps,\tau} (U^{*,\tau},V^{*,\tau})$. 
        \item[3)] (Existence of minimizers) there exists a unique minimizer $\Gamma^{*,\tau} \in \DM(\cH_1 \otimes \cH_2)$ for the unbalanced primal problem, i.e. $\primtau = \primaltau(\Gamma^{*,\tau})$, and it is given by 
        \begin{align}
            \Gamma^{*,\tau}
                =
            \psi' \left[ \frac{U^{*,\tau} \oplus V^{*,\tau} - C}{\varepsilon} \right]
                ,
        \end{align}
        where $(U^{*,\tau},V^{*,\tau})$ are the unique maximizers for the unbalanced dual problem. 
    \end{itemize}
\end{theo}

Note that the unbalanced case is characterized by the uniqueness of maximizers in the dual problem, in contrast with the balanced case. 

\subsubsection*{General duality theorem}
The assumptions on $\psi$ are classical: in particular, the strict convexity of $\psi$ ensures the uniqueness (up to the trivial translation) of the maximizers of the dual problem. The smoothness assumption $\psi\in C^1$ ensures that $\varphi$ is strictly convex in its domain, as well as it allows us to write the explicit formula for $\Gamma^*$ as given in \eqref{eq:main_formula_minimizer}. 
On the other hand, the validity of the duality result does not require such regularity but holds true for general convex, superlinear, and bounded from below $\psi$, which automatically follows from the general assumptions on $\varphi$. 
\begin{theo}[Duality for nonsmooth $\psi$]
\label{thm:duality_general_intro}
     Let $\varphi:[0,+\infty) \to [0,+\infty)$ be a convex, superlinear, and bounded from below. Then for every $\tau_1$, $\tau_2 \in [0,+\infty]$, the primal and dual problem coincide, i.e. $\dualu  = \primtau$. If $\rho \in \DM(\cH_1)$ and $\sigma \in \DM(\cH_2)$, the same holds true for the balanced case, i.e. $\dual  = \prim$.
\end{theo}
It would be interesting to investigate the validity of a formula in the same style as \eqref{eq:main_formula_minimizer} for nonsmooth $\psi$, possibly involving the subdifferential of $\psi$. We leave this to future investigations. 

\subsubsection*{Convergence results}
Having at disposal a complete understanding of both the quantum optimal transport problem and its unbalanced relaxation, our final contribution is the study of the limit behavior of both primal and dual problems as the trade-off parameters $\tau_1$, $\tau_2$ goes to $+\infty$. In particular, we show that not only do the unbalanced primal and dual functionals  $\Gamma$-converge to the corresponding ones in the balanced case, but we also provide the convergence of the associated optimizers (up to a suitable renormalization for the dual problem).

\begin{theo}[From unbalanced to balanced optimal transport: convergence result]
\label{thm:main_convergence}
    Under the same assumptions of Theorem~\ref{thm:duality_general_intro}, we have the following convergence results:
    \begin{enumerate}
        \item The functional $\primaltau$  $\Gamma$-converge as $\tau_1,\tau_2 \to \infty$ to the functional 
        \begin{align}
            \Gamma \mapsto 
            \begin{cases}
                \primal(\Gamma)
                    &\text{if } \Gamma \mapsto (\rho,\sigma) 
            \\
                +\infty 
                    &\text{otherwise} 
                        \, .
            \end{cases} 
            %     \qquad 
            % \big( \text{resp. } - \emph{D}^\eps  \big)
                % .
        \end{align}
        Similarly for the dual problem, we have that $-\emph{D}^{\eps,\tau} $ $\Gamma$-converge to $-\emph{D}^\eps $.
        % $\primal$ with marginals constraint $\Gamma \mapsto (\rho,\sigma)$ (resp. $\emph{D}^\eps $).
        \item The unique minimizer $\Gamma^{*,\tau} \in \rmH_\geq(\cH_1 \otimes \cH_2)$ of $\primaltau$ converges (up to subsequence) as $\tau_1,\tau_2 \to +\infty$ to a minimizer $\Gamma^* \in \rmH_\geq(\cH_1 \otimes \cH_2)$ of $\primal$ with constraint $\Gamma^*  \mapsto (\rho,\sigma)$. 
         \item The unbalanced optimal transport problems converge as $\tau_1,\tau_2 \to +\infty$  to the balanced optimal transport problem, i.e.
        \begin{align}
        \label{eq:conv_duality_intro}
            \lim_{\tau_1,\tau_2 \to +\infty} 
                \primtau
                    =
                \prim
            %         \quad  \tand  \quad  
                    =
                \dual 
                    =
            \lim_{\tau_1,\tau_2 \to +\infty} 
                \dualu 
                    .
        \end{align}
        \item Let $(U^{*,\tau}, V^{*,\tau}) \in \rmH(\cH_1) \times \rmH(\cH_2)$ be the unique maximizer of $\emph{D}^{\eps,\tau} $. Define the recentered potentials 
        \begin{align}
            \big(
                \hat U^{*,\tau}, \hat V^{*,\tau}
            \big) 
                := 
            \big( 
                U^{*,\tau} + \lambda_1(V^{*,\tau})
                    , 
                V^{*,\tau} - \lambda_1(V^{*,\tau})
            \big)
                \in \rmH(\cH_1) \times \rmH(\cH_2)
                    .
        \end{align}
        Then $(\hat U^{*,\tau}, \hat V^{*,\tau})$ converges (up to subsequence) as $\tau_1,\tau_2 \to +\infty$ to a maximizer $(U^*,V^*) \in \rmH(\cH_1) \times \rmH(\cH_2)$ of $\emph{D}^\eps $ which satisfies $\lambda_1(V^*)=0$.
    \end{enumerate}
    % \nat{what is the difference between 2-3?}
    Assume additionally that $\psi \in C^1$ and strictly convex. Then the convergence of the minimizers and maximizers in 2. and 4. are true without taking subsequences, and $\Gamma^{*,\tau}$ (resp. $(\hat U^{*,\tau},\hat V^{*,\tau})$) converge to the unique minimizer $\Gamma^*$ (resp. unique maximizer $(U^*,V^*)$ with $\lambda_1(V^*)=0$).
\end{theo}
% {\color{black} Add a final statement about the convergence of the transforms.  I would write it for smooth $\psi$, and only later below when we prove it, we may add a remark about the fact that a similar statement holds true also for nonsmooth $\psi$, where ofc the transform is not really defined (but one can use any sequence of maximizers, possibly not unique). }
% 
Finally, the final convergence result we obtain is about the convergence of the $(C,\psi,\eps,
\tau)$-transforms associated with the dual functional of the unbalanced optimal transport problem.
% We introduce the following notation: for $(U,V) \in \rmH(\cH_1) \times \rmH(\cH_2)$, we set 
%     \begin{align}
%     \label{eq:def_decomposition_terms_intro}
%     \begin{gathered}
%         \text{F}_1(U) = \tau_1 \Tr[e^{-\frac{U}{\tau_1}+\log \rho}] 
%             ,   \qquad
%         \text{F}_2(V) = \tau_2 \Tr[e^{-\frac{V}{\tau_2}+\log \sigma}]
%             ,
%     \\[6pt]
%         \text{G}(U,V) = \ep \Tr\left[\psi\left(\frac{U\oplus V -C}{\ep}\right)\right] 
%             .
%     \end{gathered}
%     \end{align}

The corresponding transform in the unbalanced case is given by the map $\mathscr{F}_2^{(C,\psi,\eps,\tau)} \colon \rmH(\cH_1)\to \rmH(\cH_2)$  such that
\begin{equation}
\label{eq:U_transform_ubal_intro}
    \mathscr{F}_2^{(C,\psi,\eps,\tau)}(U) 
        =
    \Argmax_{V \in \rmH(\cH_2)}
     \left\{ 
        -  \tau_2 \Tr[e^{-\frac{V}{\tau_2}+\log \sigma}] - \ep \Tr\left[\psi\Big(\frac{U\oplus V -C}{\ep}\Big)\right] 
            % \suchthat
    \right\}.
\end{equation}
Analogously, we define $\mathscr{F}_1^{(C,\psi,\eps,\tau)} \colon \rmH(\cH_2)\to \rmH(\cH_1)$ as
\begin{equation}
\label{eq:V_transform_ubal_intro}
    \mathscr{F}_1^{(C,\psi,\eps,\tau)}(V) 
        = 
    \Argmax_{U \in \rmH(\cH_1)}
    \left\{ 
        -  \tau_1 \Tr[e^{-\frac{U}{\tau_1}+\log \rho}] - \ep \Tr\left[\psi\Big(\frac{U\oplus V -C}{\ep}\Big)\right] 
            % \suchthat
    \right\}.
\end{equation}

\begin{theo}[Convergence of $C$-transforms]
    Assume that $\psi$ is strictly convex. Then we have that the associated $(C,\psi,\eps,\tau)$-transforms converge, as the parameter $\tau \to +\infty$. More precisely, one has that, for $i=1,2$, 
\begin{align}
    \mathscr{F}_i^{(C,\psi,\eps,\tau)}(W_\tau)
        \to 
    \mathscr{F}_i^{(C,\psi,\eps)}(W_\infty)
        \qquad 
    \forall \, \, 
        W_\tau \to W_\infty 
            .
\end{align}
\end{theo}

A similar result could be proved as well without assuming that $\psi$ is strictly convex. Of course, in this case, the definition of $(C,\psi,\eps)$-transform is not well-defined, due to the fact that we don't have a unique maximizer in \eqref{eq:U_transform_bal_intro}, \eqref{eq:V_transform_bal_intro}. Nonetheless, it not hard to prove that unbalanced $C$-transforms will converge (up to subsequence) to a maximizer in \eqref{eq:U_transform_bal_intro} (resp. \eqref{eq:V_transform_bal_intro}).

\begin{rem}[Alternative methods]
One may wonder if it is possible to prove Theorem \ref{thm:main_duality_balanced} only using the results in Theorem \ref{thm:main_duality_unbalanced} by means of an approach via $\Gamma$-convergence. The main point that should be addressed is to prove the last equality in the following expression
\begin{equation*}
     \lim_{\tau_1,\tau_2 \to +\infty} 
                \primtau
                    =
                \prim
            %         \quad  \tand  \quad  
                    =
                \dual 
                    =
            \lim_{\tau_1,\tau_2 \to +\infty} 
                \dualu 
\end{equation*}
by simple arguments.
Indeed, we prove in Theorem \ref{thm:main_convergence} that $\dual 
                    =
            \lim_{\tau_1,\tau_2 \to +\infty} 
                \dualu $ as a consequence of the weak duality of balanced problem.
The main challenge is: although the fundamental theorem of $\Gamma$-convergence can be applied on the primal problem, this is not the case for the dual, since $\{D^{\eps,\tau}\}_{\tau}$ are not equi-coercive (Remark \ref{rmk:dualcoercivity}). As a result, we are unable to establish the existence of maximizers for the (balanced) quantum optimal transport using this technique. Instead, the existence of maximizers in the dual formulation of quantum optimal transport is proved directly in Theorem \ref{thm:main_duality_balanced} via the Direct method of Calculus of Variations and a shifting argument.

\section{Quantum optimal transport}
%
% \nat{Do we define the primal problem somewhere? I think it can be usefull, because in the dual we sort of either need to pose that the function $\psi$ is convex and superlinear, OR we go from primal and say that $\psi = \psi^*$ and the Legendre transform is convex+superlinear. I know that in most situations we actually go from dual to primal (mostly relevant for infinite-dimensional) or whichever we decide.}
\label{sec:balanced_noncommutative}
In this section, we study the convex regularized quantum optimal transport problems, both in its primal and dual formulation. 
We begin by analyzing the dual problem, and discussing the existence of maximizers and their properties. We then use them to prove a duality result, namely the equality between the primal and the dual problem. In the last part of this section, we discuss the notion of $(C,\psi,\eps)$-transform associated with a general regularization $\psi$.

Throughout the whole section,  $\cH_1,~\cH_2$ denotes Hilbert spaces of finite dimensions $d_1$ and $d_2$, respectively. We fix $C \in \rmH(\cH_1 \otimes \cH_2)$ and consider $\rho \in \mkP(\cH_1)$, $\sigma \in \mkP(\cH_2)$ with trivial kernel. 
Recall the definition of the dual functional 
\begin{align}
\label{eq:ddual_balanced}
    \dualf (U,V)=\Tr\left[ U \rho\right] + \Tr\left[ V \sigma\right]-\eps \Tr \left[ \psi \Big(\frac{U\oplus V - C}{\eps} \Big) \right],
\end{align}

% \lore{
% $\eps = 1$, $A_{U,V}:=  U\oplus V - C$. 
% \begin{align}
%     \dualf (U,V)
%         = 
%     ..
% \end{align}
% }

where $\psi \in C(\R)$ is a convex, superlinear function at infinity, and bounded from below, i.e. 
\begin{align}
\label{eq:assumptions_psi}
    \lim_{t \to +\infty}
        \frac{\psi(t)}{t}
    = + \infty 
        \tand 
    m:= \inf \psi > - \infty 
        .
\end{align}
% {\color{blue} [lore]: Do we need it also at $-\infty$? is it true that it is superlinear at $-\infty$ if $\psi = \varphi^*$ with $\varphi \in C(\R_+)$ is superlinear only at $+\infty$? }
\begin{rem}[Properties of superlinear functions]
\label{rem:superlinear_prop1}
    A crucial property that we use later in this section is that if $\psi$ is superlinear, then the function $\R \ni x \mapsto x - \psi(x)$ has superlevels which are bounded from above, namely for every $\alpha \in \R$, 
    \begin{align}
        \left\{
            x \in \R 
                \suchthat 
            x - \psi(x) \geq \alpha 
        \right\}
            \subset 
        \left\{
            x \in \R 
                \suchthat 
            x \leq R_\alpha 
        \right\}
            , 
    \end{align} 
    for some $R_\alpha \in (0,+\infty)$. This directly follows from the fact that for every $R>0$, 
    \begin{align}
        \left\{
            x \in \R 
                \suchthat 
            x - \psi(x) \geq \alpha 
        \right\}  
            \cap 
        \left\{
            x \in \R 
                \suchthat 
            x > R 
        \right\}
            \subset 
        \left\{
            x \in \R 
                \suchthat 
            \frac{\psi(x)}{x} 
                \leq 
            \max 
            \Big\{ 
                1 - \frac\alpha R
                    , 
                1
            \Big\}
        \right\}  
            ,
    \end{align}
    and the fact that the set on the right-hand side is bounded by the definition of superlinearity.\fr
\end{rem}
% Let us consider $\varphi$ and $\psi$ as defined in Section \ref{sec:setting_mr} and let $C \in \rmH(\cH)$.
% By the Fenchel-Young inequality we have, given $\eps >0$
% %
% \begin{equation}\label{eq:fenchel-young}
%     \eps \langle\, \Gamma, \frac{U \oplus V -C}{\eps} \, \rangle \le \eps \psi\left(\Gamma\right) +\eps \psi^*\left(\frac{ U\oplus V - C}{\eps}\right) 
% \end{equation}
% holds for every $U \in \rmH(\cH_1), V \in \rmH(\cH_2)$ and $\Gamma \in D_1(H_1 \otimes H_2)$.
%
%
\end{rem}

\begin{rem}[Invariance by translation]
\label{rem:symmetry}
Note that the dual functional $\dualf$ enjoys special symmetry properties. In particular, for every $(U,V) \in \rmH(\cH_1) \times \rmH(\cH_2)$ we have that
% As in the classical case, for the analogous functional $D^\psi_\eps(u,v)$, in the class of $(\varphi,\psi) \in L^\psi_\eps(\mu) \times L^\psi_\eps(\mu)$, 
\begin{align}
    \dualf (U + \lambda{\rm Id}_{\cH_1} , V - \lambda {\rm Id}_{\cH_2} )=\dualf (U, V)
\end{align}
for every $\lambda \in \mathbb{R}$. This in particular shows that $\dualf $ is not coercive. Nonetheless, as shown at the end of this section, this is the unique invariance of the problem, see Remark~\ref{rem:unique_symmetry}. 
\fr
\end{rem}
% {\color{blue} [Lore]: uniqueness has to be done, maybe later cause we are not ready.}
% 
% Note that the assumption on the kernel of $\rho$ and $\sigma$ is crucial to this purpose: indeed 
% Given $(U,V) \in \rmH(\cH_1) \times \rmH(\cH_2)$, we compute a class of matrices $(A,B) \in \rmH(\cH_1) \times \rmH(\cH_2)$ for which $D_\eps^\psi(U,V)=D_\eps^\psi(U+A,V+B)$.
% Indeed, we have
% \begin{equation}
% \begin{aligned}
%     0&= D_\eps^\psi(U+A,V+B) - D_\eps^\psi(U,V)\\
%     &= \Tr\left[ A \rho \right]+\Tr\left[ B \sigma \right]+\eps \Tr \left[ \psi \left(\frac{U \oplus V + A \oplus B - C}{\eps} \right) - \psi \left(\frac{U \oplus V - C}{\eps} \right) \right].\\
% \end{aligned}
% \end{equation}
% In particular, given $\lambda \in \mathbb{R}$, this holds if $A=\lambda {\rm Id}_{\cH_1}$ and $B= -\lambda {\rm Id}_{\cH_2}$.
%
\subsection{Existence of maximizers}
\label{sec:max_exist}
% the proof of the existence of maximizers and duality -- Nataliia notes.

First of all, we observe that the maximization problem
\begin{align}
    \sup \{\dualf  (U,V) \suchthat U\in \rmH(\cH_1),\, V\in \rmH(\cH_2) \},
\end{align}
is equivalent, due to the symmetry argument in Remark \ref{rem:symmetry}, to the restricted maximization 
\begin{equation}\label{eq:duar_renormalized}
    \sup \{\dualf  (U,V) \suchthat U\in \rmH(\cH_1),\, V\in \rmH(\cH_2), \, \lambda_1(V)=0 \},
\end{equation}
where recall that $\lambda_1(U)$ denotes the smallest eigenvalue of $U$.

The next result shows that once restricted to the class of operators $(U,V)$ with $\lambda_1(V) = 0$, the dual functional is coercive.
% {
% \color{black}
%     [Lore]: I kept essentially the same proof as Nata's, just using directly the sum and making it a little bit shorter. 
% }

\begin{proposition}[Coercivity]
\label{prop:superlevel_bdd}
%Let $\cH_1,~\cH_2$ be Hilbert spaces of dimensions $d_1$ and $d_2$, respectively. 
% Let $\rho\in H_\geq(\cH_1)$ and $\sigma\in H_\geq(\cH_2)$ be definite positive density matrices on $H_1$ and $H_2$, respectively, and let $\psi:\R\to\R$ be a superlinear convex function bounded from below by a constant $m\in\R$. 
For every  $U\in \rmH(\cH_1)$, $V\in \rmH(\cH_2)$ such that 
\begin{align}
    \lambda_1(V)=0 
        \tand  
    \emph{D}^\eps (U,V)\geq M >-\infty 
\end{align}
for some $M\in\R$, there exists a constant $0\leq \mathcal{A}=\mathcal{A}(M) <\infty$ such that $||U||_\infty,||V||_\infty\leq \cA$.
\end{proposition}
% \nat{We need to type a lemma about superlinearity of the Legendre transform -- [Lore: done]}
\begin{proof}
For any given $U\in \rmH(\cH_1)$, $V\in \rmH(\cH_2)$, we denote by $S:= U \oplus V \in\rmH(\cH_1 \otimes \cH_2)$ and write its spectral decomposition as $S= \sum_{k=1}^d \lambda_k(S) |\xi_k \rangle \langle \xi_k |$ where $d= d_1 d_2$. Note that every eigenvalue of $S$ is given by the sum of an eigenvalue of $U$ and one of $V$, hence to show that the $||U||_\infty,~||V||_\infty$ norms are bounded it suffices to provide an upper and a lower bound the spectrum of $S$, since
\begin{align}
    \lambda_1(S) = \lambda_1(U) \leq \lambda_{d_1}(U)
        \ \  \text{and} \ \    
    \lambda_d(S) = \lambda_{d_1}(U)+\lambda_{d_2}(V) 
        \geq 
            \max
            \{ \lambda_{d_1}(U) , \lambda_{d_2}(V) + \lambda_1(S) \} 
        .
\end{align}
% $\lambda_1(S) = \lambda_1(U) \leq \lambda_{d_1}(U)$ and $\lambda_d(S) = \lambda_{d_1}(U)+\lambda_{d_2}(V) \geq \lambda_{d_1}(U)$.

Writing the dual functional in terms of $S$, one can apply the Peierl's inequality (e.g., Theorem 2.9 in \cite{carlen2010trace}) for the base $\{\ket{\xi_k}\}_{k=1}^{d}$ and obtain 
% Jensen's inequality (e.g., Theorem 1.2 in \cite{PecFurMicSeo05})
\begin{align}
\label{eq:finite_energy_bound}
    M 
        &\leq 
    \dualf (U,V) 
        = 
    \Tr[S (\rho\otimes\sigma)] 
        -
    \eps \Tr\left[\psi\Big(\dfrac{S-C}{\ep}\Big)\right] 
\\
        &\leq  
    \sum_{k=1}^{d} 
    \left(                      
        \bra{\xi_k} S(\rho\otimes\sigma) \ket{\xi_k} 
        -
            \ep \psi
        \Big(
            \bra{\xi_k} \dfrac{S - C}{\eps} \ket{\xi_k}
        \Big)
    \right) 
\\
        &=
    \sum_{k=1}^{d} 
    \left( 
        \lambda_k(S) 
        \bra{\xi_k} \rho\otimes\sigma \ket{\xi_k} 
            -
        \ep \psi
        \Big(
            \dfrac{\lambda_k(S) - \bra{\xi_k} C \ket{\xi_k}}{\eps} 
        \Big)
    \right)
\\
        &= 
    \sum_{k=1}^{d} 
    \left(
        \Lambda_k \omega_k 
            -
        \ep \psi
        \Big(
            \dfrac{\Lambda_k - C_k}{\eps} 
        \Big)
    \right),
\end{align}
where for simplicity we set $\Lambda_k:= \lambda_k(S)$, $\omega_k = \bra{\xi_k} \rho\otimes\sigma \ket{\xi_k}$, and $C_k=\bra{\xi_k} C \ket{\xi_k}$. 

On the one hand, using that $\sum_{k=1}^{d}\omega_k=1$, $\Lambda_1 \leq \Lambda_d$, as well as $\omega_1 \geq \lambda_1(\rho \otimes \sigma)$, we have that
\begin{align*}
    M 
        &\leq
    \Lambda_1 \omega_1 
        + 
    \sum_{k=2}^d 
        \Lambda_k \omega_k 
            -
        \eps 
        \sum_{k=1}^d 
            \psi
            \Big(
                \dfrac{\Lambda_k - C_k}{\eps} 
            \Big) 
% \\
    \leq  \Lambda_1 \omega_1 + \Lambda_d ( 1-\omega_1) -\eps d m
\\
    &\leq 
    \Lambda_1 \lambda_1(\rho \otimes \sigma ) + \Lambda_d (1-\lambda_1(\rho \otimes \sigma )) -\ep d m
        , 
\end{align*}
% {\color{red} [update here]}\
which yields the lower bound $\Lambda_1 \geq \lambda_1(\rho \otimes \sigma )^{-1} \big( M+ \eps d m - \Lambda_d (1-\lambda_1(\rho \otimes \sigma )) \big)$, where we used that $\lambda_1(\rho \otimes \sigma ) > 0$ due to the fact that $\rho$ and $\sigma$ have trivial kernel. 

Note that $\lambda_1(\rho \otimes \sigma) \in (0,1]$, therefore we are left to show that $\Lambda_d$ is bounded from above, which follows by superlinearity of $\psi$. Indeed, from \eqref{eq:finite_energy_bound} we see that
\begin{align}
    M 
        &\leq  
    \Lambda_d  
        - 
    \eps \psi
    \Big(
        \dfrac{ \Lambda_d - C_d}{\eps} 
    \Big) 
        -
    \eps 
    \sum_{k=1}^{d-1} 
        \psi
        \Big(
            \dfrac{\Lambda_k - C_k}{\eps} 
        \Big) 
% \\
%     \leq & \Lambda_{d_1 d_2} - C_{d_1 d_2} + C_{d_1 d_2} - \ep \psi\left(\dfrac{\Lambda_{d_1 d_2} - C_{d_2 d_2}}{\ep} \right) -\ep m (d_1 d_2 -1) 
\\
        &\leq 
    \eps 
    \left(
         \dfrac{ \Lambda_d - C_d}{\eps} 
            - 
        \psi
        \Big(
             \dfrac{ \Lambda_d - C_d}{\eps}
        \Big)
    \right) 
        +
    \| C \|_{\infty} - \ep m (d-1)
        ,
\end{align}
which yields
\begin{align}
     % \left(
         \dfrac{ \Lambda_d - C_d}{\eps} 
            - 
        \psi
        \Big(
             \dfrac{ \Lambda_d - C_d}{\eps}
        % \right)
    \Big) 
        \geq 
    \dfrac{M+\ep m (d_1 d_2 -1) - ||C||_{\infty}}{\ep}
        .
\end{align}
The conclusion the directly follows from the superlinearity of $\psi$ (see in particular Remark~\ref{rem:superlinear_prop1}) and the fact that $|C_d| \leq \| C \|_\infty$. 
\end{proof}

% \nat{Please take a look where is the best place to put this "remark".}
\begin{rem}\label{rem:shift}
    Notice that without loss of generality, due to the symmetry of the functional $\dualf $, the condition $\lambda_1(V)=0$ can be replaced with  $\lambda_1(U)=0$. In fact, the claimed compactness works with any other constraint of the form $\lambda_1(U)\in B$ (or $V$ instead of $U$), for a given bounded set $B \subset \R$. \fr
\end{rem}

\begin{rem}[Dependence of $\cA$ on the cost and dimension]
    Following the proof of the latter Proposition, one can see that $\cA$ depends on $\| C \|_\infty$, and the dimensions $d_1$, $d_2$ of the underlying Hilbert spaces. Nonetheless, it is clear from the proof that the second dependence disappears whenever $\psi\geq 0$, as well as in the limit as $\eps \to 0$, which suggests the possibility of (partially) extending these results to infinite dimensional setups.\fr
\end{rem}

% \nat{Here we need a Proposition that $D^\psi_\eps$ is continuous.}
% \lore{Done}

We may now proceed to show the existence of the maximizer for the dual functional \eqref{eq:def_dual_balanced}.
\begin{theo}[Existence of the maximizer]
\label{thm:max_exist}
Let $\varepsilon>0,$ $\cH_1,~\cH_2$ be finite-dimensional Hilbert spaces.
Let $\rho\in \mkP(\cH_1)$ and $\sigma\in \mkP(\cH_2)$ with trivial kernel and let $\psi:\R\to\R$ be a superlinear convex function bounded from below. Then the dual functional $\emph{D}^\eps $ defined in \eqref{eq:def_dual_balanced} admits a maximizer $(U^*,V^*)\in \rmH(\cH_1)\times\rmH(\cH_2)$.
\end{theo}
\begin{proof}
    The proof follows from the direct method.
    % \lore{Not necessary, by coercivity the limit will always have finite energy}
    % First, see that due to \eqref{eq:fenchel-young}, $\sup D^\psi_\eps(\cdot,\cdot)$ has a finite value. 
    Indeed, take a maximizing sequence $\{(U^n,V^n)\}_{n\geq 1} \subset \rmH(\cH_1)\times\rmH(\cH_2)$ such that  
    % \begin{align}
    $
        \sup 
        \dualf 
            = 
        \lim_{n \to \infty} 
            \dualf (U^n,V^n)
                .
    $
    % \end{align}
    Thanks to the observation in \eqref{eq:duar_renormalized}, we can assume $\lambda_1(V^n) = 0$. In particular, we then have 
    \begin{align}
        \{(U^n,V^n)\}_{n\geq 1}
            \subset 
        \{(U,V)\in \rmH(\cH_1)\times\rmH(\cH_2) \suchthat \lambda_1(V)=0, \, \dualf (U,V) \geq M \},
    \end{align}
    for some $M \in \R$. Due to Proposition~\ref{prop:superlevel_bdd}, we may conclude that the sequence $\{(U^n,V^n)\}_{n\geq 1}$ is bounded, and we extract a converging subsequence $(U^{n_k},V^{n_k})\xrightarrow[]{k\to\infty} (U^*,V^*)\in \rmH(\cH_1)\times\rmH(\cH_2) $. 
    By the continuity of $\dualf $ \eqref{eq:dualf_continuous}, taking the limit as $k \to \infty$ we obtain that $\dualf (U^*,V^*) = \sup \dualf $, hence a maximizer.
    % 
    % Finally, we have that $\sup \dualf  \leq \dualf (U^{n_k},V^{n_k}) + \frac{1}{n_k}$, and by a standard argument, taking $\limsup$ we get
    % \[
    %     \sup \dualf  
    %         \leq 
    %     \limsup_{k\to\infty} 
    %         \dualf (U^{n_k},V^{n_k}) + \frac{1}{n_k} = \dualf  (U^*,V^*)
    %         \leq 
    %     \sup \dualf (U,V),
    % \]
    % which yields that $(U^*,V^*)$ is a maximizer.    
\end{proof}

\subsection{\texorpdfstring{$(C,\psi,\eps)-$transform,}{} properties of maximizers, and duality}
\label{sec:Cepspsi_transform}
% \nat{Reminder to check the indexes of transforms and so on.}
% \lore{they should be ok now}
%
% We investigate some basic properties of the $(C,\psi,\eps)$ transform.
Recall the setting introduced at the beginning of Section~\ref{sec:balanced_noncommutative}.
% Given $C \in \rmH(\cH)$, $\eps >0$ and $\psi$ as above, 
We now introduce the notion of the $(C,\psi,\eps)$-transform, in analogy with the regularization given by the Shannon entropy in the quantum case in \cite{FelGerPor23}. 
% In this part of the section, we additionally assume that $\psi$ is {\color{black} strictly convex}.
%
\begin{defin}[$(C,\psi,\eps)$-transform]
Assume additionally that $\psi$ is {\color{black} strictly convex}. We define $\mathscr{F}_2^{(C,\psi,\eps)} \colon \rmH(\cH_1)\to \rmH(\cH_2)$ as
\begin{equation}\label{eq:U_transform_bal}
    \mathscr{F}_2^{(C,\psi,\eps)}(U) 
        =
    \argmax \left\{ \Tr[ V \sigma] - \eps \Tr \left[ \psi \Big(\frac{U\oplus V - C}{\eps} \Big) \right]
        \suchthat 
    V \in \rmH(\cH_2) \right\}.
\end{equation}
Analogously, we define the corresponding $\mathscr{F}_1^{(C,\psi,\eps)} \colon \rmH(\cH_2)\to \rmH(\cH_1)$ as
\begin{equation}
\label{eq:V_transform_bal}
    \mathscr{F}_1^{(C,\psi,\eps)}(V) 
        = 
    \argmax 
    \left\{ 
        \Tr[ U \rho] - \eps \Tr \left[ \psi \Big(\frac{U\oplus V - C}{\eps} \Big) \right] 
            \suchthat 
        U \in \rmH(\cH_1) 
    \right\}.
\end{equation}
\end{defin}
The definition of the transform is well-posed, as the $\argmax$ contains exactly one element. Indeed, due to Remark \ref{rem:shift}, we may provide the proof only for $\mathscr{F}_1^{(C,\psi,\eps)}$. For every given $V\in \rmH(\cH_2)$ we have that
\begin{align*}
    &\argmmax
    \limits_{U\in \rmH(\cH_1)} 
    \left\{ 
        \Tr[ U \rho] - \eps \Tr \left[ \psi \Big(\frac{U\oplus V - C}{\eps} \Big) \right]
    \right\}
        =  
    \argmmax\limits_{U\in \rmH(\cH_1)} 
     \left\{ 
        \dualf (U,V) 
    \right\}
\\
        = 
    &\argmmax\limits_{U\in \rmH(\cH_1)} \left \{ \dualf (U + \lambda_1(V) {\rm Id}_{\cH_1},V - \lambda_1(V) {\rm Id}_{\cH_2}) \right \}.
\end{align*}
In particular, as we are looking for a maximizer, we can restrict the maximization set to 
\begin{equation}\label{eq:set_U_maximize}
    \left \{ U\in \rmH(\cH_1) \suchthat \dualf (U + \lambda_1(V) {\rm Id}_{\H_1},V - \lambda_1(V) {\rm Id}_{\cH_2})\geq  M \right \} ,
\end{equation}
for some $M \in \R$.
Under this condition, we have that $(U + \lambda_1(V) {\rm Id}_{\H_1},V - \lambda_1(V) {\rm Id}_{\cH_2})$ satisfies the assumption of the Proposition~\ref{prop:superlevel_bdd}, and therefore the set \eqref{eq:set_U_maximize} is bounded. Arguing as in Theorem \ref{thm:max_exist}, one can show that a maximizer in \eqref{eq:V_transform_bal} exists.
% $$\mathscr{F}_1^{(C,\psi,\eps)}(V)\in \argmmax\limits_{U\in \rmH(\cH_1)} \left \{ \Tr[ U \rho] - \eps \Tr \left[ \psi \left(\frac{U\oplus V - C}{\eps} \right) \right] \right \}.$$
% Uniqueness, we observe that $\dualf (\mathscr{F}_1^{(C,\psi,\eps)}(V) , V) \geq \dualf (U, V)$. On the other hand, due to Theorem 2.10 in \cite{carlen2010trace}, since $\psi$ is strictly convex, $\Tr[\psi(\cdot)]$ is strictly convex as well, and thus the maximizer of $\dualf (\cdot,V)$ is unique.
Uniqueness follows from the strict concavity of $\dualf$, as a consequence of \eqref{eq:dualf_continuous}.
\begin{proposition}[Optimality conditions for the $(C,\psi,\eps)$-transforms]
\label{prop:optimality_conditions} 
    Under the standing assumptions of this section, let additional suppose that $\psi$ is {\color{black} strictly convex and $C^1$}. Then, for every $U \in \cH_1$, its transform $\mathscr{F}_2^{(C,\psi,\eps)}(U)$ is the unique solution $\overline V$ of the equation 
\begin{equation}
\label{eq:optimality_conditions}
    \sigma =\tP_2 \left[ \psi'\Big( \frac{U \oplus \overline{V} - C}{\eps} \Big)\right].
\end{equation}
Analogously, the characterization holds for $\mathscr{F}_1^{(C,\psi,\eps)}(V)$, replacing $\sigma$ with $\rho$ and $\tP_2$ with $\tP_1$. If $\psi$ is only $C^1$ and convex (not necessarily strictly), then transform $\mathscr{F}_2^{(C,\psi,\eps)}(U)$ is not uniquely determined, but any (possibly not unique) maximizer of \eqref{eq:U_transform_bal} satisfies the same optimality condition \eqref{eq:optimality_conditions}.
\end{proposition}
\begin{proof}
    Set $V:= \mathscr{F}_2^{(C,\psi,\eps)}(U)$ (or in general any maximizer for \eqref{eq:U_transform_bal}). For every $\delta \in \mathbb{R}$ and $A \in \rmH(\cH_2)$, we define the map
    \begin{align}
        f(\delta):=\ep \Tr\left[ \psi \Big(\frac{U \oplus (V+\delta A)-C}{\ep} \Big) \right] - \Tr\left[ (V+\delta A) \sigma\right]
            .
    \end{align}
    By optimality, we know that $f'(0)=0$, and by Lemma~\ref{lem:gradient} this translated into 
    \begin{align}
        \Tr\left[ \psi' \Big(\frac{U \oplus V-C}{\ep} \Big)({\rm Id}\otimes A) \right] - \Tr\left[ A \sigma\right] = 0
            \, , \quad 
        \forall A \in \rmH(\cH_2)
            ,
    \end{align}
   which provides the sought characterization, by the very definition of partial trace $\tP_2$. The second statement follows by arguing similarly. Uniqueness comes directly from the strict concavity of the function $V \mapsto \dualf (U,V)$ whenever $\psi$ is strictly convex.
\end{proof}
\begin{exam}[Quadratic regularization]
\label{rem:optimality_condition_t2_balanced}
A typical choice of regularization which differs from the usual von Neumann entropy is given by the quadratic case $\varphi(t)=\frac12 t^2$, $t \geq 0$. See for example \cite{LorManMey21} for the study of the regularized optimal transport problem with such $\psi$, including numerical methods to seek for maximizers. They are based on the fact that in this case, the associated Legendre transform is given by $\psi(t) = \frac12 (t_+)^2$. 
The optimality conditions therefore read as 
\begin{equation*}
    \sigma 
        =
    \tP_2 
    \left[ 
    \left(
        \frac{U \oplus \mathscr{F}_2^{(C,\psi,\eps)}(U) - C}{\eps} 
    \right)_+
    \right]
    %     =
    % \frac{1}{\varepsilon} 
    % \left( 
    %     \Tr[U] {\rm Id}_2 + d_2 \mathscr{F}_2^{(C,\psi,\eps)}(U)-\tP_2[C]  
    % \right)
        .
\end{equation*}
Note this is an example of $\psi \in C^1$ which is not strictly convex. In particular, $\psi'$ is not globally invertible and an explicit solution cannot be found, and the maximizers of the dual problem are not necessarily determined by $\lambda_1(U)$ (cfr. \cite[Section 2.3]{LorManMey21} for the commutative setting). 
% giving the explicit formula
% \begin{equation*}
%     \mathscr{F}_2^{(C,\psi,\eps)}(U)
%         =
%     \frac1{d_2} \big( \varepsilon\sigma+\tP_2[C]-\Tr[U] {\rm Id}_2 \big) 
%         .
% \end{equation*} 
\fr
\end{exam}

\begin{theo}[Equivalent characterizations for maximizers]
\label{thm:characterization_maximizers}
    Under the standing assumptions of this section, let $(U^*, V^*) \in \rmH(\cH_1) \times \rmH(\cH_2)$. Assume that $\psi$ is {\color{black} strictly convex and $C^1$}. Then the following are equivalent:
    \begin{itemize}
        \item[1)] (Maximizers) $(U^*,V^*)$ maximizes $\emph{D}^\eps $, i.e. $\dual  = \emph{D}^\eps (U^*,V^*)$.
        \item[2)] (Maximality condition) $U^* = \mathscr{F}_1^{(C,\psi,\eps)}(V^*)$ and $V^* = \mathscr{F}_2^{(C,\psi,\eps)}(U^*)$.
        \item[3)] (Complementary slackness) Let 
        $ \displaystyle 
        \Gamma^*:=\psi' \left[ \frac{U^* \oplus V^* - C}{\varepsilon} \right]$, then $\Gamma^* \mapsto (\rho,\sigma)$.
        % \item[4)] (Duality \emph{I}) The operator $\Gamma^*$ defined in 3) is such that  $\primal(\Gamma^*)= \emph{D}^\eps (U^*,V^*)$.
        \item[4)] (Duality) There exists $\Gamma^* \in \rmH(\cH_1 \otimes \cH_2)$ so that $\Gamma^* \mapsto (\rho, \sigma)$ and $\primal(\Gamma^*)= \emph{D}^\eps (U^*,V^*)$.
    \end{itemize}
If one (and thus all) condition holds, $\Gamma^*$, as defined in 3) is the unique minimizer to $\prim$. 
\end{theo}

As an immediate corollary of Theorem~\ref{thm:characterization_maximizers} and Theorem~\ref{thm:max_exist}, we obtain the sought duality.
\begin{coro}[Duality for convex regularized QOT]\label{cor:duality_balanced}
    Under the standing assumptions of this section, assume that $\psi = \varphi^*$ is strictly convex and $C^1$. Then $\dual=\prim$. 
\end{coro}
\begin{proof}[Proof of Theorem~\ref{thm:characterization_maximizers}]
    We shall show 1) $\Rightarrow$ 2) $\Rightarrow$ 3) $\Rightarrow$ 4) $\Rightarrow$ 1). 

    \smallskip
    \noindent 
    1) $\Rightarrow$ 2). \ 
    By definition of $\mathscr{F}_1^{(C,\psi,\eps)}(V^*)$, we have that $\dualf (\mathscr{F}_1^{(C,\psi,\eps)}(V^*),V^*) \ge \dualf (U^*,V^*) =\dual $. 
    % This observation, together with the assumption gives that $\dualf (\mathscr{F}_1^{(C,\psi,\eps)}(V^*),V^*) = \dualf (U^*,V^*)$. 
    By uniqueness of maximizer of $\dualf (\cdot,V^*)$, we have that $U^*=\mathscr{F}_1^{(C,\psi,\eps)}(V^*)$. A similar argument proves the second part of the statement.

    \smallskip 
    \noindent 
    2) $\Rightarrow$ 3). \ 
   This holds by Proposition \ref{prop:optimality_conditions}.

    \smallskip 
    \noindent 
    3) $\Rightarrow$ 4). \ 
    Denote by simplicity $W^*:=(U^* \oplus V^*- C)/\varepsilon$, and compute
    \begin{align}
    \label{eq:main_computation_equivalence}
        \dualf (U^*,V^*)
            &= 
        \Tr
        \left[
            U^* \rho\right]+\Tr\left[V^* \sigma\right] -\varepsilon \Tr\left[\psi\Big(\frac{U^*\oplus V^*-C}{\varepsilon}\Big)
        \right]
    \\
            &\stackrel{3)}{=} 
        \Tr
        \left[
            (U^*\oplus V^*) \Gamma^*\right] -\varepsilon \Tr\left[\psi\Big(\frac{U^*\oplus V^*-C}{\varepsilon}\Big)
        \right]
    \\
            &=
        \varepsilon
        \Tr
        \left[ 
            W^* \psi'(W^*)-\psi(W^*)\right]+\Tr\left[C \Gamma^*
        \right]
        %     =
        % \varepsilon
        % \Tr\left[ \widehat{g}(W^*) \right]+\Tr\left[C \Gamma^* 
        % \right]
            ,
    \end{align}
    %where $g(x):=x \psi'(x)-\psi(x)$.
    As a consequence of \eqref{eq:legendre_equality}, we have that
    \begin{align}
    \label{eq:application_lemma}
    \Tr
        \left[ 
            W^* \psi'(W^*)-\psi(W^*)
        \right]
            =
        \Tr\left[  \varphi \big(  \psi'(W^*) \big) \right]= \Tr\left[ \varphi(\Gamma^*)\right].
    \end{align}
    Therefore, we continue the computation in \eqref{eq:main_computation_equivalence} using \eqref{eq:application_lemma} and we have
    \begin{align}
        \dualf (U^*,V^*) =\Tr\left[C \Gamma^* \right] +\varepsilon \Tr\left[ \varphi(\Gamma^*)\right]
        = \primal(\Gamma^*).
    \end{align}
    This proves 4). 
    
    \smallskip 
    \noindent 
    4) $\Rightarrow$ 1). \     
    Since $\primal (\Gamma^*) = \dualf (U^*,V^*)\le \prim$ by Remark~\ref{rem:primal_bigger_dual}, we have that $\Gamma^*$ is a minimizer for the primal problem, hence $\dualf (U^*,V^*)= \prim$. 
    Moreover, once again by Remark~\ref{rem:primal_bigger_dual} it holds $\prim \ge \dualf (U,V)$ for every $(U,V) \in \rmH(\cH_1) \times \rmH(\cH_2)$, hence we also deduce that $\dualf (U^*,V^*) \ge \dualf (U,V)$ for every $(U,V) \in \rmH(\cH_1) \times \rmH(\cH_2)$, which precisely shows that $(U^*,V^*)$ is a maximizer.
\end{proof}
\begin{rem}[Maximizers are translation of each other]
    \label{rem:unique_symmetry}
    As we have seen in Remark~\ref{rem:symmetry}, the dual functional is invariant by translations that sum up to zero. Whenever $\psi$ is {\color{black} strictly convex and $C^1$}, we claim that every pair of maximizers can be obtained by translation of one another. More precisely, given $(U_1,V_1)$ , $(U_2,V_2) \in \rmH(\cH_1) \times \rmH(\cH_2)$ two pairs of maximizers, then there must exist $\lambda \in \R$ so that $(U_1,V_1) = (U_2+ \lambda {\rm Id}_{\cH_1}, V_2 - \lambda {\rm Id}_{\cH_1})$.  To see this, we first observe that $\psi$ being smooth implies that $\varphi$ is strictly convex on the interior of its domain, see e.g. \cite[Thm~26.3]{Rockafellar:1970}. In particular, the primal functional admits a unique minimizer, hence thanks to Theorem~\ref{thm:characterization_maximizers}, we deduce that
    \begin{align}
        \psi' 
        \left[ 
            \frac{U_1 \oplus V_1 - C}{\varepsilon} 
        \right]
            =
        \psi' 
        \left[ 
            \frac{U_2 \oplus V_2 - C}{\varepsilon} 
        \right]
            . 
    \end{align}
    By strict convexity, we know that $\psi'$ is injective, hence we infer that 
    \begin{align}
    \label{eq:equality_sum}
        % \left[ 
            \frac{U_1 \oplus V_1 - C}{\varepsilon} 
        % \right]
            =
        % \left[ 
            \frac{U_2 \oplus V_2 - C}{\varepsilon} 
        % \right]
            \quad \Rightarrow \quad 
        U_1 \oplus V_1
            =
        U_2 \oplus V_2
            .
    \end{align}
    For simplicity, for $\lambda \in \R$ we here adopt the  shorthand notation $\lambda := \lambda {\rm Id}_{\cH_i}$. By computing the marginals from \eqref{eq:equality_sum}, we see that 
    \begin{align} 
    \begin{cases}
        U_1 d_2 + d_1\Tr(V_1) = U_2 d_2 + d_1 \Tr(V_2) 
            \, ,
    \\
        \Tr(U_1) d_2 + d_1 V_1 = \Tr(U_2) d_2 + d_1 V_2
            \, ,
    \end{cases}   
        \hspace{-2mm} 
        \Rightarrow \, 
    \begin{cases}
        U_1 - U_2 
            = 
        \frac{d_1}{d_2} 
        \big( 
            \Tr(V_2) - \Tr(V_1) 
        \big) 
            =: \lambda_1
            \, ,
    \\
        V_1 - V_2 
            = 
        \frac{d_2}{d_1} 
        \big( 
            \Tr(U_1) - \Tr(U_2) 
        \big)
            =: \lambda_2
            \, .
    \end{cases}
    \end{align}
      The fact that $\lambda_1 = -\lambda_2$ clearly follows by using these relations together with \eqref{eq:equality_sum}, due to the fact that $(\lambda_1 {\rm Id}_{\cH_1}) \oplus (\lambda_2 {\rm Id}_{\cH_2} ) = ( \lambda_1 + \lambda_2 ) {\rm Id}_{\cH_1 \otimes \cH_2} $.\fr
\end{rem}
\subsection{Sinkhorn iterations and their convergence}
In this section, we define the Sinkhorn iterations associated with a general regularization $\psi$ and show it provides a maximizing sequence for the dual functional $\dualf $ defined in \eqref{eq:def_dual_balanced}. Thus, we prove the qualitative convergence of the iteration.
\begin{proposition}\label{prop:sinkhorn}
% Let $\cH_1,~\cH_2$ be Hilbert spaces of dimensions $d_1$ and $d_2$, respectively. 
Under the standing assumptions of the section, assume additionally that {\color{black} $\psi$ is strictly convex and $C^1$}. Let $\mathscr{F}_1^{(C,\psi,\eps)}:\rmH(\cH_2)\to\rmH(\cH_1)$ and $\mathscr{F}_2^{(C,\psi,\eps)}:\rmH(\cH_1)\to\rmH(\cH_2)$ be defined as in \eqref{eq:V_transform_bal} and \eqref{eq:U_transform_bal}, respectively. We define the following algorithm: for $n \in \N$, 

\underline{\emph{Step 0}}: fix any initial $(U^0,V^0)\in \rmH(\cH_1)\times\rmH(\cH_2)$. 

\underline{\emph{Step 2n-1}}: for given $U^{n-1} \in \rmH(\cH_1)$, we define $V^n$ as 
\begin{align}
    V^n := \mathscr{F}_2^{(C,\psi,\eps)}(U^{n-1}) - \lambda_1(\mathscr{F}_2^{(C,\psi,\eps)}(U^{n-1})){\rm Id}_{\cH_2}
        .
\end{align}

\underline{\emph{Step 2n}}: we define $U^n$ as 
\begin{align}
     U^n := \mathscr{F}_1^{(C,\psi,\eps)}(\mathscr{F}_2^{(C,\psi,\eps)}(U^{n-1})) + \lambda_1(\mathscr{F}_2^{(C,\psi,\eps)}(U^{n-1})){\rm Id}_{\cH_1}.
\end{align}
% For $(U^0,V^0)\in \rmH(\cH_1)\times\rmH(\cH_2)$ 
% % be such that $\lambda_1(V) = 0$ 
% and $\dualf (U^0, V^0)\geq M$ for some $M\in \R$, define 
% \begin{equation}
% \label{eq:sinkhorn}
% \begin{cases}
% % \displaystyle 
%     V^n = \mathscr{F}_2^{(C,\psi,\eps)}(U^{n-1}) - \lambda_1(\mathscr{F}_2^{(C,\psi,\eps)}(U^{n-1})){\rm Id}_{\cH_2},
% \\
% % \displaystyle 
%     U^n = \mathscr{F}_1^{(C,\psi,\eps)}(\mathscr{F}_2^{(C,\psi,\eps)}(U^{n-1})) + \lambda_1(\mathscr{F}_2^{(C,\psi,\eps)}(U^{n-1})){\rm Id}_{\cH_1}.
% \end{cases}
% \end{equation}
Then we have the convergence of the iteration. More precisely:
\begin{enumerate}
    \item[(a)] There exists a limit point $(U^*,V^*)$ of  $\{(U^n,V^n)\}_{n\geq 0}$, which is a maximizer of $\emph{D}^\eps $.
    % \item[(b)] There exists a sequence $\{\mu^n\}_{n\geq 1}\subset \R$ and $(U^*,V^*)\in \rmH(\cH_1)\times\rmH(\cH_2)$ such that 
    % \[
    %     (U^n+\mu^n,V^n-\mu^n) \xrightarrow[]{n\to\infty} (U^*,V^*),
    % \]
    % and $(U^*, V^*)$ is a maximizer of $\dualf $.
    \item[(b)] The operator 
    $\displaystyle 
        \Gamma^n := \psi'
        \Big(
            \frac{U^n\oplus V^n - C}{\ep}
        \Big)
    $ belongs to $\mkP(\cH_1 \otimes \cH_2)$, and moreover we have that $\displaystyle \Gamma^n \to \Gamma^* \in \mkP(\cH_1 \otimes \cH_2)$ as $n \to \infty$, where 
    \begin{align}
        \Gamma^* 
            = 
        \psi'
        \Big( 
            \frac{U^* \oplus V^* - C}{\ep}
        \Big)
            \mapsto 
        (\rho,\sigma) 
    \end{align}
    is a minimizer of the primal problem $\prim$.
\end{enumerate}
\end{proposition}
\begin{rem}[Sinkhorn iteration]
    The iteration is defined in the following way. We first start with an arbitrary initialization $(U^0,V^0)$. Then, we compute iteratively $(U^n,V^n)$ as follows. Denote by $\tilde{V}^n$ the $(C,\psi,\varepsilon)$-transform of $U^{n-1}$. On one hand, 
    $V^n$ is obtained recentering $\tilde{V}^n$ so that $\lambda_1(V^n)=0$. 
    On the other hand, $U^n$ is obtained by computing the $(C,\psi,\varepsilon)$-transform of $\tilde{V}^n$ and translating it so that $\dualf (U^n,V^n) = \dualf  (\mathscr{F}_1^{(C,\psi,\varepsilon)}(\tilde{V}^n),\tilde{V}^n)$. Intuitively, this procedure enjoys good compactness properties by Proposition~\ref{prop:superlevel_bdd}, by ensuring that $\lambda_1(V^n)=0$, for every $n \in \N$. Finally, the fact that limit points are maximizers for the dual functional is suggested by the properties of the $(C,\psi,\eps)$-transform. Note indeed that by the optimality conditions \eqref{eq:optimality_conditions} 
    \begin{align}
        \tP_1
        \Big[
            \psi'
            \Big(
                \frac{U^n\oplus V^n - C}{\ep}
            \Big)
        \Big]
            = 
        \rho 
            \quad \text{whereas} \quad 
        \tP_2
        \Big[
            \psi'
            \Big(
                \frac{U^n\oplus \tilde V^{n+1} - C}{\ep}
            \Big)
        \Big]
            = 
        \sigma
            , 
    \end{align}
    % \nat{why is there tilde?}
    and compare with 3) in Theorem~\ref{thm:characterization_maximizers}.\fr
\end{rem}
\begin{proof}
    We set $M:= \dualf (U_0,V_0) \in \R$. 

    \smallskip 
    \noindent
    \textit{Proof of} (a). \ 
    See that, by construction, for any $n\geq 0$ we have $\lambda_1(V^n) = 0$ and additionally, due to the invariance by translation (Remark~\eqref{rem:symmetry}), from the definition of $\mathscr{F}_1^{(C,\psi,\eps)}$ and the one of $\mathscr{F}_2^{(C,\psi,\eps)}$, we obtain that
    \begin{align*}
        \dualf (U^n,V^n)  =&  \dualf ( \mathscr{F}_1^{(C,\psi,\eps)}(\mathscr{F}_2^{(C,\psi,\eps)}(U^{n-1})) , \mathscr{F}_2^{(C,\psi,\eps)}(U^{n-1}) ) \\
        \geq &  \dualf ( U^{n-1} , \mathscr{F}_2^{(C,\psi,\eps)}(U^{n-1}) ) \\
        \geq & \dualf ( U^{n-1} , V^{n-1} ) \geq \dots \geq \dualf ( U^{0} , V^{0} ) = M.
    \end{align*}
    We conclude that the sequence $\{(U^n,V^n)\}_{n\geq 0}$ satisfies the assumption of Proposition~\ref{prop:superlevel_bdd}, and therefore is bounded. Thus, there exist an accumulation point $(U^*, V^*)\in \rmH(\cH_1)\times\rmH(\cH_2)$ and a subsequence $(U^{n_k}, V^{n_k})\xrightarrow[]{k\to\infty} (U^*, V^*)$. Now, using the continuity of $\dualf $:
    \begin{align*}
        \dualf (U^*,V^*) =& \limsup_{k\to\infty} \dualf (U^{n_k}, V^{n_k})  \geq \limsup_{k\to\infty} \dualf (U^{n_{k-1}+1}, V^{n_{k-1}+1}) \\
        =& \limsup_{k\to\infty}\dualf ( \mathscr{F}_1^{(C,\psi,\eps)}(\mathscr{F}_2^{(C,\psi,\eps)}(U^{n_{k-1}})) , \mathscr{F}_2^{(C,\psi,\eps)}(U^{n_{k-1}}) ) \\
        \geq & \limsup_{k\to\infty} \dualf ( U^{n_{k-1}} , \mathscr{F}_2^{(C,\psi,\eps)}(U^{n_{k-1}}) ) \\
        \geq & \limsup_{k\to\infty} \dualf ( U^{n_{k-1}} , \mathscr{F}_2^{(C,\psi,\eps)}(U^{*}) ) = \dualf ( U^{*} , \mathscr{F}_2^{(C,\psi,\eps)}(U^{*}) ),
    \end{align*}
    where in the last inequality, we used once again that $\mathscr{F}_2^{(C,\psi,\eps)}(U^{n_{k-1}})$ is a maximizer of $V \mapsto \dualf (U^{n_{k-1}},V)$.
    This, and a similar computation for $U^*$, implies that 
    \begin{equation}\label{eq:limit_transform}
        V^* = \mathscr{F}_2^{(C,\psi,\eps)}(U^{*}), \quad\quad \text{and} \quad\quad U^* = \mathscr{F}_1^{(C,\psi,\eps)}(V^{*}),
    \end{equation} 
    hence $(U^*, V^*)$ is a maximizer due to Theorem~ \ref{thm:characterization_maximizers}. 
    In addition, one observes that $\lambda_1(V^*)= 0$, as if $A_k \to A$ then $\lambda_1(A_k) \to \lambda_1(A)$ as $k \to \infty$. 
    % {\color{black} [honestly I would not even write the proof, in any form, it is well-known -- e.g. \cite[Corollary 2.5]{Kuwae-Shioya:2003}, in a muuuch more general setting]
    % % 
    % This follows from the fact that the functional 
    % % \begin{align} 
    % $
    %     \ket{\psi} \mapsto \langle \psi | A_k | \psi \rangle 
    % $
    % % \end{align}
    % is $\Gamma$-converging to the same functional with $A$ instead of $A_k$. Using that $\lambda_1(V) = \min \{ \bra{\psi} V \ket{\psi} \suchthat ||\ket{\psi}||=1\}$ 
    % and that the constraint $\| \ket{\psi} \| = 1$ is closed, we conclude.
    % }
    Let $(U^{n_j}, V^{n_j})$ be any converging subsequence of $(U^n, V^n)$ and denote by $(\tilde{U}, \tilde{V})$ their limit point. For the first part of the proof, we obtain that \eqref{eq:limit_transform} also holds for $(\tilde{U}, \tilde{V})$,  thus it is also a maximizer. On the other hand, due to Remark \ref{rem:symmetry}, we must then have $(\tilde{U},\tilde{V}) = (U^*+\lambda \rm{Id}_{\cH_1}, V^*-\lambda \rm{Id}_{\cH_2})$ for some $\lambda \in \R$. However, $\lambda_1(\tilde{V})=0$ as well, and thus $\lambda=0$ or $(\tilde{U},\tilde{V}) = (U^*, V^*)$.

    % \nat{not relevant anymore}
    % (b) We will explicitly construct the sequence $\{\mu^n\}_{n\geq 0}$. Consider the subsequence $\{(U^{n_k}, V^{n_k})\}_{k\geq 1}$ from part (a) and define 
    % \begin{equation}
    %     \mu^n = \begin{cases}
    %         0, & n=n_k \\
    %         \argmin \{||(V^n - \mu{\rm Id}_{\cH_2}) - V^*||_{\infty}, \suchthat \mu\in \R\}, & n\neq n_k.
    %     \end{cases}
    % \end{equation}
    % Suppose that $||(V^n - \mu^n{\rm Id}_{\cH_2}) - V^*||_{\infty}\not\to 0$. Then for some $\delta>0$ one can find a subsequence $||(V^{n_j} - \mu^{n_j}{\rm Id}_{\cH_2}) - V^*||_{\infty}> \delta$. From (a) follows that the sequence $\{(U^{n_j}, V^{n_j})\}_{j\geq 1}$ is bounded, thus one can extract a further converging subsequence and find an accumulation point \nat{CHECK AGAIN...}

    % , and by construction we will have  $||(V^{n_j} - \mu{\rm Id}_{\cH_2}) - V^*||_{\infty}> \delta$ for arbitrary $\mu\in \R$.

    \smallskip 
    \noindent 
    \textit{Proof of }(b). \
    % \nat{Please, fill}
    It follows by the optimality conditions that $\tP_1
        \Big[
            \psi'
            \Big(
                \frac{U^n\oplus V^n - C}{\ep}
            \Big)
        \Big]
            = 
        \rho$, therefore, by taking the trace at both sides we get that $\Tr[\Gamma^n]=1$, thus proving $\Gamma^n \in \mkP(\cH_1 \otimes \cH_2)$. The convergence $\Gamma^n \to \Gamma$ follows from (a) and the fact that $\psi'$ is continuous.
\end{proof}

\section{Unbalanced quantum optimal transport}
\label{sec:unbalanced_noncommutative}
In this section, we turn our attention to the quantum optimal transport problems, both the primal and the dual, in the unbalanced case. 
We begin by analyzing the dual problem, discussing the existence of maximizers.
We work under the same standing assumptions of Section \ref{sec:balanced_noncommutative}, namely throughout the whole section,  $\cH_1,~\cH_2$ denotes Hilbert spaces of finite dimensions $d_1$ and $d_2$, respectively. We fix $C \in \rmH(\cH_1 \otimes \cH_2)$ and consider $\rho \in \rmH_>(\cH_1)$, $\sigma \in \rmH_>(\cH_2)$ (in particular with trivial kernel). 

\subsection{Existence of the maximizers to the unbalanced dual problem }
In this section, we prove the existence of maximizers in the dual problem for the unbalanced case for a general regularization.
We recall the definition of the dual functional $\dualfu $
\begin{equation}
\begin{aligned}
    \dualfu (U,V)
    &= -\tau_1
     \Tr
     \left[ 
        e^{-\frac{U}{\tau_1}+\log \rho}- \rho
    \right] 
        -
    \tau_2
    \Tr
    \left[ 
       e^{-\frac{V}{\tau_2}+\log \sigma}- \sigma
    \right] 
        -\eps 
    \Tr
    \left[ 
        \psi\Big(\frac{U\oplus V-C}{\eps}\Big)
    \right]
        .
    % \\
    % &=-\text{F}_1(U)-\text{F}_2(V)-\text{G}(U,V)+\tau_1 \Tr[\rho] +\tau_2 \Tr[\sigma].\\
\end{aligned}
\end{equation}
In order to show the existence of maximizers, we show directly that the functional $ \dualfu $ is coercive, differently from what happens in the balance case with $ \dualf $ (due to the presence of the symmetries, cfr. Remark~\ref{rem:symmetry}). Nonetheless, we significantly simplify the proof by observing that the unbalanced dual functional $\dualfu$ is in fact always bounded from above by $\dualf$, as we show in the proof below.
\begin{theo}[Existence of the maximizer, unbalanced]
\label{thm:max_exist_unbalanced}
    Let $\cH_1,~\cH_2$ be finite-dimensional Hilbert spaces.
    Let $\rho\in \rmH_>(\cH_1)$ and $\sigma\in \rmH_>(\cH_2)$ with trivial kernel, $\varepsilon, \tau_1$, $\tau_2 >0$ and let $\psi:\R\to\R$ be a superlinear convex function bounded from below. Then the unbalanced dual functional $\emph{D}^{\eps,\tau} $ defined in \eqref{eq:def_dual_unbalanced} admits a (unique) maximizer $(U^*,V^*)\in \rmH(\cH_1)\times\rmH(\cH_2)$.
\end{theo}

\begin{proof}
    First of all, we observe that the functional $ \dualfu $ is continuous over $\rmH(\cH_1) \times \rmH(\cH_2)$, arguing in the same way as in \eqref{eq:dualf_continuous}. Therefore, the existence of a maximizer follows by showing that $ \dualfu $ is coercive, i.e. it has bounded superlevel sets. Moreover, without loss of generality, we can assume that $\psi \geq 0$.
    
    We introduce the following notation: for $(U,V) \in \rmH(\cH_1) \times \rmH(\cH_2)$, we set 
    \begin{align}
    \label{eq:def_decomposition_terms}
    \begin{gathered}
        \text{F}_1(U) = \tau_1 \Tr[e^{-\frac{U}{\tau_1}+\log \rho}] \geq 0
            ,   \qquad
        \text{F}_2(V) = \tau_2 \Tr[e^{-\frac{V}{\tau_2}+\log \sigma}]\geq 0
            ,
    \\[6pt]
        \text{G}(U,V) = \ep \Tr\left[\psi\Big(\frac{U\oplus V -C}{\ep}\Big)\right] \geq 0
            ,
    \end{gathered}
    \end{align}
    so that by construction, for every $(U,V) \in \rmH(\cH_1) \times \rmH(\cH_2)$ one has 
    \begin{align}
    \label{eq:decomposition_unbalanced}
        \dualfu (U,V) 
            = 
        -\text{F}_1(U) - \text{F}_2(V) - \text{G}(U,V)
            + \tau_1 \Tr[\rho] + \tau_2 \Tr[\sigma] 
        .
    \end{align}
    We start by claiming that 
    \begin{align}
    \label{eq:dualf>dualfu}
         \dualf (U,V)  \geq   \dualfu (U,V) 
            \, , \quad \text{for all } (U,V) \in \rmH(\cH_1) \times \rmH(\cH_2)
            \, , \tau_1, \tau_2 \in \R_+ 
                .
    \end{align}
    This is consequence of Klein's inequality \cite[Thm.~2.11]{carlen2010trace} applied to the exponential function, which implies (for F$_1$, and similarly for F$_2$)
    \begin{align}
     \text{F}_1(U)
        - \tau_1 \Tr[\rho]
        =
    \tau_1
     \Tr
     \left[ 
        e^{-\frac{U}{\tau_1}+\log \rho}- e^{\log \rho}
    \right]
        % \stackrel{\text{KI}}{\geq}
        \geq 
    \tau_1
     \Tr
     \left[ 
       e^{\log \rho} 
       \Big(
            - \frac{U}{\tau_1}
       \Big)
    \right]
        = 
    - \Tr [ U \rho ] 
        , 
    \end{align}
    thus providing the claimed lower bound \eqref{eq:dualf>dualfu}.

    Now, we want to show that, given $M_0 \in \mathbb{R}$, the set
    \begin{equation}
        \mathscr{S}_{M_0}:= \{ (U,V) \in \rmH(\cH_1) \times \rmH(\cH_2): \, \dualfu (U,V) \ge - M_0 \}
    \end{equation}
    is bounded. In particular, up to a translation that depends only on the marginals, there exists $M\in \R$ so that $\mathscr{S}_{M_0}$ is a subset of 
    \begin{align}
       \{ (U,V) \in \rmH(\cH_1) \times \rmH(\cH_2): \, \text{F}_1(U)+\text{F}_2(V)+\text{G}(U,V) \leq M \}.
    \end{align}
    For $(U,V) \in \mathscr{S}_M$, we have that $0\leq \text{F}_1(U)+\text{F}_2(V) + \text{G}(U,V) \le M$, which, by the positivity of all the terms, gives that $\text{F}_1(U) \le M$, $\text{F}_2(V) \le M$, $\text{G}(U,V) \le M$. 
    % Next, we can bound from above and below the eigenvalues of $U$ and $V$. 

    We define $\hat{U}:=-\frac{U}{\tau_1}+\log \rho$.
    We write the spectral decomposition $\hat{U} =\sum_i \lambda_i\ket{\xi_i} \bra{\xi_i}$, where $\lambda_i := \lambda_i(\hat{U})$. Therefore
    \begin{align}
        M 
            \geq
        \text{F}_1(U) 
            = 
        \tau_1 \Tr[e^{\hat{U}}] 
            = 
        \tau_1 \sum_i e^{\lambda_i}  
            \geq \tau_1 e^{\lambda_d}
            \quad\Rightarrow
            \quad
        \lambda_{d_1} \le \log\Big(\frac{M}{\tau_1}\Big).
    \end{align}
     %where $\rho_d$, which is the smallest eigenvalue of $\rho$, is positive, since $\rho$ is definite positive. 
     Since $U=-\tau_1 \hat{U}+\tau_1\log\rho$, we have $\lambda_1(U)\ge\lambda_1(-\tau_1 \hat{U})+\lambda_1(\tau_1\log\rho)=-\tau_1\lambda_{d_1}(\hat{U})+\tau_1\log(\lambda_1(\rho))$.
     Combining the previous inequalities, we conclude that 
    \begin{equation}
    \label{eq:U_1_lower}
        \lambda_1(U) 
            \geq 
        -\tau_1 \log \Big( \frac{M}{\lambda_1(\rho)\,\tau_1} \Big)
            =: \cA_1 > -\infty 
                .
    \end{equation}
    Arguing analogously for $V$ and $\text{F}_2(V)$, we also obtain that
    \begin{equation}
    \label{eq:V_1_lower}
        \lambda_1(V) 
            \geq 
        -\tau_2 \log \Big( \frac{M}{\lambda_1(\sigma)\,\tau_2} \Big)
            =: \cA_2 > -\infty 
                .
    \end{equation}
    In order to obtain the sought upper bound on the spectrum of $U$ and $V$, we take advantage of \eqref{eq:dualf>dualfu}, which implies $\dualf (U,V) \geq M_0$. In particular, we apply Proposition~\ref{prop:superlevel_bdd} and we get $||U+\lambda_1(V)\rm{Id}_{H_1}|| \leq \cA$ and $||V-\lambda_1(V)\rm{Id}_{H_2}|| \leq \cA$, hence in particular for all $i=1,\dots,d_1$ and $j=1,\dots,d_2$, the two-sided bounds
    \begin{align}
        -\cA 
            \leq 
        \lambda_i(U) +\lambda_1(V) 
            \leq 
        \cA
            , \quad \tand \quad 
        0
            \leq 
        \lambda_j(V) -\lambda_1(V) 
            \leq 
        \cA
            .
    \end{align}
    Along with \eqref{eq:U_1_lower} and \eqref{eq:V_1_lower} one easily obtains 
    \begin{align}
        \cA_1
            \leq 
        \lambda_i(U)
            \leq 
        \cA - \cA_2
            , \quad \tand \quad 
        \cA_2 
            \leq 
        \lambda_j(V)
            \leq 
        2\cA - \cA_1
            ,
    \end{align}
    which yields $\mathscr{S}_{M_0}$ is bounded, thus the sought coercivity property for $\dualfu $.
\end{proof}

\subsection{The \texorpdfstring{$(C,\psi,\eps,\tau)$}{}-transform and its properties}
%
% We investigate some basic properties of the $(C,\psi,\eps)$ transform.
Recall the setting introduced at the beginning of Section~\ref{sec:unbalanced_noncommutative}.
% Given $C \in \rmH(\cH)$, $\eps >0$ and $\psi$ as above, 
Similarly as done in Section~\ref{sec:Cepspsi_transform}, we now introduce the corresponding notion of $(C,\psi,\eps)$-transform associated to $\dualfu $. Note that in this case, we do not require $\psi$ to be strictly convex, due to the presence of the exponential (which is strictly convex) in the marginal penalization terms. In particular, the existence of the maximizer in the definition below follows from the concavity and continuity of $\dualfu $, together with the coercivity proved within the proof of Theorem~\ref{thm:max_exist_unbalanced}, whereas uniqueness comes exactly from the strict concavity of $\dualfu $. 

For simplicity of notation, we set $\tau := (\tau_1,\tau_2)$. Recall the decomposition as given in \eqref{eq:decomposition_unbalanced}, in particular the definitions of F$_1$, F$_2$, and G in \eqref{eq:def_decomposition_terms}.

% In this part of the section, we additionally assume that $\psi$ is {\color{black} strictly convex}.
%
\begin{defin}[$(C,\psi,\eps,\tau)$-transform]
% Assume additionally that $\psi$ is {\color{black} strictly convex}. 
We define $\mathscr{F}_2^{(C,\psi,\eps,\tau)} \colon \rmH(\cH_1)\to \rmH(\cH_2)$ as
\begin{equation}\label{eq:U_transform_ubal}
    \mathscr{F}_2^{(C,\psi,\eps,\tau)}(U) 
        =
    \argmax 
    \left\{ 
        - \text{F}_2(V) - \text{G}(U,V)
            \suchthat 
        V \in \rmH(\cH_2) 
    \right\}.
\end{equation}
Analogously, we define the corresponding $\mathscr{F}_1^{(C,\psi,\eps,\tau)} \colon \rmH(\cH_2)\to \rmH(\cH_1)$ as
\begin{equation}
\label{eq:V_transform_ubal}
    \mathscr{F}_1^{(C,\psi,\eps,\tau)}(V) 
        = 
    \argmax 
    \left\{ 
        - \text{F}_1(U) - \text{G}(U,V)
            \suchthat 
        U \in \rmH(\cH_2)
    \right\}.
\end{equation}
\end{defin}
The definition of the transform is well-posed, as the $\argmax$ contains exactly one element, as mentioned above. The next step is to characterize the optimality condition in the same spirit as in Proposition~\ref{prop:optimality_conditions}.
\begin{proposition}[Optimality conditions for the $(C,\psi,\eps,\tau)$-transforms]
\label{prop:optimality_conditions_unbalanced} 
    Under the standing assumptions of this section, assume additionally that $\psi \in C^1$. Then, for every $U \in \rmH(\cH_1)$, its transform $\mathscr{F}_2^{(C,\psi,\eps,\tau)}(U)$ is the unique solution $\overline V$ of the equation 
\begin{equation}
\label{eq:optimality_conditions_unbalanced}
    e^{-\frac{\overline V}{\tau_2}+\log \sigma} =\tP_2 \left[ \psi'\Big( \frac{U \oplus \overline{V} - C}{\eps} \Big)\right].
\end{equation}
Analogously, the characterization holds for $\mathscr{F}_1^{(C,\psi,\eps,\tau)}(V)$, replacing $\sigma$ with $\rho$, $\tP_2$ with $\tP_1$.
\end{proposition}

\begin{proof}
    Set $V:= \mathscr{F}_2^{(C,\psi,\eps,\tau)}(U)$. For every $\delta \in \mathbb{R}$ and $A \in \rmH(\cH_2)$, arguing as in the proof of Proposition~\ref{prop:optimality_conditions}, from the fact that the map $\delta \mapsto - \text{F}_2(V+\delta A) - \text{G}(U,V+\delta A)$ has vanishing derivative in $\delta =0$, we obtain that 
    \begin{align}
        \Tr
        \left[ 
            e^{-\frac{V}{\tau_2} + \log \sigma} A
        \right]
            -
        \Tr
        \left[ 
            \psi' \Big(\frac{U \oplus V-C}{\ep} \Big)({\rm Id}\otimes A) 
        \right] 
            = 
        0
            \, , \quad 
        \forall A \in \rmH(\cH_2)
            ,
    \end{align}
   which provides the sought \eqref{eq:optimality_conditions_unbalanced}, by the very definition of partial trace $\tP_2$. The second statement follows by arguing similarly. Uniqueness comes directly from the strict concavity of the function $V \mapsto \dualfu (U,V)$.
\end{proof}
%
% LORE: for now I left it out, tbd.
% 
% \begin{rem}[Optimality condition for $\psi(t)=t^2/2$, unbalanced]
% \label{rem:optimality_condition_t2_unbalanced}
% \lore{UPDATE}
% The discussion in Example~\ref{rem:optimality_condition_t2_balanced} in this case takes the following, more complicated, form. 
% Applying Proposition \ref{prop:optimality_conditions_unbalanced} to this specific case, we get that $\bar{V}:=\mathscr{F}_2^{(C,\psi,\eps)}(U)$ solves
% \begin{equation*}
%     e^{-\frac{\overline V}{\tau_2}+\log \sigma}= \frac{1}{\varepsilon} \left( \Tr[U] {\rm Id}_{\mathcal{H}_2} + d_2 \bar{V}-\tP_2[C]  \right)
% \end{equation*}
% Rearranging the terms and defining $\hat{V}:=-\frac{\bar{V}}{\tau_2}+\log \sigma$, we have
% \begin{equation*}
%     e^{\hat{V}}+\frac{\tau_2 d_2}{\varepsilon} \hat{V} =\frac{\Tr[U]{\rm Id}_{\mathcal{H}_2} - \tP_2[C]}{\varepsilon} +\frac{\tau_2 d_2}{\varepsilon}\log \sigma =:W.
% \end{equation*}
% Since the function $h(x):=e^x + \tau_2 \frac{d_2}{\eps} x$ is invertible and denoting by ${\rm inv}$ its inverse we have $\hat{V}={\rm inv}(W)$.
% Therefore, computing numerically ${\rm inv}$ allows us to compute $\hat V$, and hence $\bar{V}$.\fr
% \end{rem}

Recall the definition of the primal functional \eqref{eq:def_primal_unbalanced} and problem \eqref{eq:def_primal_problem_unbalanced} in the unbalanced case.

\begin{theo}[Equivalent characterizations for maximizers, unbalanced]
\label{thm:characterization_maximizers_unbalanced}
    Under the standing assumptions of this section, assume $\psi \in C^1$. Let $(U^*, V^*) \in \rmH(\cH_1) \times \rmH(\cH_2)$. Then the following are equivalent:
    \begin{itemize}
        \item[1)] (Maximizers) $(U^*,V^*)$ maximizes $\emph{D}^{\eps,\tau} $, i.e. $\dualu  = \emph{D}^{\eps,\tau} (U^*,V^*)$.
        \item[2)] (Maximality condition) $U^* = \mathscr{F}_1^{(C,\psi,\eps,\tau)}(V^*)$ and $V^* = \mathscr{F}_2^{(C,\psi,\eps,\tau)}(U^*)$.
        % \item[3)] (Complementary slackness) Let 
        % $ \displaystyle 
        % \Gamma^*:=\psi' \left[ \frac{U^* \oplus V^* - C}{\varepsilon} \right]$, then $\Gamma^* \mapsto (\rho,\sigma)$.
        % \item[4)] (Duality \emph{I}) The operator $\Gamma^*$ defined in 3) is such that  $\primal(\Gamma^*)= \emph{D}^\eps (U^*,V^*)$.
        \item[3)] (Duality) There exists $\Gamma^* \in \rmH_\geq(\cH_1 \otimes \cH_2)$ so that $\primaltau(\Gamma^*)= \emph{D}^{\eps,\tau} (U^*,V^*)$.
    \end{itemize}
If one (and thus all) condition holds, the unique $\Gamma^*$ which satisfies 3) is given by 
\begin{align}
\label{eq:Gamma_star}
    \Gamma^*:=\psi' \left[ \frac{U^* \oplus V^* - C}{\varepsilon} \right]
        ,
\end{align}
and it is the unique minimizer to $\primtau$. 
\end{theo}
As in Corollary~\ref{cor:duality_balanced}, from Theorem~\ref{thm:characterization_maximizers_unbalanced} and Theorem~\ref{thm:max_exist_unbalanced}, we obtain the sought duality.
\begin{coro}[Duality for unbalanced QOT]
    Under the standing assumptions of this section, assume that $\psi = \varphi^*$ is $C^1$. Then $\dualu=\primtau$. 
\end{coro}
\begin{proof}[Proof of Theorem~\ref{thm:characterization_maximizers_unbalanced}]
    We shall show 1) $\Rightarrow$ 2) $\Rightarrow$ 3) $\Rightarrow$ 1). 

    \smallskip 
    \noindent 
    1) $\Rightarrow$ 2) . \ 
   This follows as in the proof of Theorem~\ref{thm:characterization_maximizers}.

    \smallskip 
    \noindent 
    2) $\Rightarrow$ 3). \ 
    Consider $\Gamma^*$ as given in \eqref{eq:Gamma_star}. 
    From the properties of the Legendre transform and \eqref{eq:def_Legendre_relative_entropy} (by computing the optimality condition in the definition of $\cE^*(\cdot|\eta)$), we infer that 
    \begin{align}
        \cE(\alpha|\eta) = \Tr 
        \left[  
            \eta - e^{A+\log \eta}
        \right] 
            + \langle A, \alpha \rangle
            \, , \qquad \text{if }
                \alpha = e^{A+\log \eta}
                    .
    \end{align} 
    By applying the above equation to $\eta= \rho$ (resp.\ $\eta = \sigma$), $A = -\frac{U^*}{\tau_1}$ (resp.\ $A=-\frac{V^*}{\tau_2}$), noting that $\alpha = \tP_1(\Gamma^*)$ (resp.\ $\alpha=\tP_2(\Gamma^*)$) from the optimality condition Proposition~\ref{prop:optimality_conditions_unbalanced}, we obtain 
    \begin{align}
    \begin{cases}
    \displaystyle
        \tau_1 \cE(\tP_1(\Gamma^*)|\rho) 
            = 
        \tau_1 \Tr 
        \left[  
            \rho - e^{ -\frac{U^*}{\tau_1}+\log \rho } 
        \right] 
            - 
         \Tr
        \left[
            U^* \tP_1(\Gamma^*) 
        \right]
                    .
    \\[10pt]
    \displaystyle
        \tau_2 \cE(\tP_2(\Gamma^*)|\sigma) 
            = 
        \tau_2 \Tr 
        \left[  
            \sigma - e^{-\frac{V^*}{\tau_2}+\log \sigma} 
        \right] 
            - 
        \Tr
        \left[
            V^* \tP_2(\Gamma^*)
        \right] 
                    .        
    \end{cases}
    \end{align}
    On the other hand, the same computation as in \eqref{eq:main_computation_equivalence} shows that 
    \begin{align}
        \Tr [ C \Gamma^*] 
            +
        \eps \Tr
        \left[ 
            \varphi(\Gamma^*)
        \right]
            = 
        \Tr
        \left[
            U^* \tP_1(\Gamma^*) 
        \right]
            +
        \Tr
        \left[
            V^* \tP_2(\Gamma^*)
        \right] 
            - 
        \varepsilon \Tr
        \left[
            \psi\Big(\frac{U^*\oplus V^*-C}{\varepsilon}\Big)
        \right]
            .
    \end{align}
    Putting the last three equalities together we end up with
    \begin{align}
        \primaltau(\Gamma^*) 
            &=
        \tau_1 \Tr 
        \left[  
            \rho - e^{ -\frac{U^*}{\tau_1}+\log \rho } 
        \right]      
            + 
        \tau_2 \Tr 
        \left[  
            \sigma - e^{-\frac{V^*}{\tau_2}+\log \sigma} 
        \right] 
            -
        \varepsilon \Tr
        \left[
            \psi\Big(\frac{U^*\oplus V^*-C}{\varepsilon}\Big)
        \right]
    \\
            &=
        \dualfu (U^*,V^*)
            , 
    \end{align}
    which is the sought equality. 
    
    \smallskip 
    \noindent 
    3) $\Rightarrow$ 1). \     
    The proof is completely analogous to the one provided in the proof of Theorem~\ref{thm:characterization_maximizers}, including the fact that $\Gamma^*$ as given in \eqref{eq:Gamma_star} is the unique minimizer for the primal problem.
\end{proof}
%
% {\color{blue} [LP]: we don't do Sinkhorn for unbalanced, do we?}

\section{Extensions of the duality theorem}
\label{sec:extension_to_nonsmooth}
% \nat{should we make it as a section, not a subsection of unbalanced?} \lore{I see your point, but honestly i don't like to have like 7 sections. Maybe i prefer it like this, as a tool towarsds the convergence results. BUt i am open to discussions.}
% \lore{ Here below a possible strategy to extend the duality result to less regular $\psi$.} 
We finally show that the duality result holds with the weaker assumption of $\psi$ being (not necessarily strictly) convex, superlinear, and bounded from below. In particular, no additional smoothness is required.

\begin{theo}[Duality for non-smooth $\psi$]
\label{thm:duality_general}
    Let $\psi$ be a convex function satisfying \eqref{eq:assumptions_psi}. Under the standing assumptions of Section~\ref{sec:unbalanced_noncommutative}, duality holds, i.e. 
    \begin{align}
        \primtau 
            = 
        \dualu 
            \, , 
    \end{align}
    for every $\tau \in (0,+\infty]$ (where $\tau = +\infty$ denotes the balanced problem).
\end{theo}

Before proving the theorem, we discuss some preliminary constructions. Let $\psi$ be a superlinear function, bounded from below, and convex (but not necessarily strictly convex nor $C^1$). 
Fix a parameter $n \in \N$, a smooth, nonnegative, monotone increasing and strictly convex function  $\bar \psi$, and a function $\rho \in C_c^\infty(\R)$, strictly positive, symmetric, and supported on $[-1,1]$, and such that $\int \rho \dd x = 1$. Construct the sequence of mollifiers $\rho_n$ in the usual way, i.e. $\rho_n(x) := n \rho(nx)$, so that $\rho_n$ satisfies the same properties of $\rho$ (on the interval $[-\frac1n, \frac1n]$), for every $n \in \N$. Finally, define
\begin{align}
    \psi_n: \R \to \R 
        \, , \quad 
    \psi_n := 
    \Big(
        \psi + \frac1n \bar \psi
    \Big)
        *
    \rho_n
         .
\end{align}
Note that in the unbalanced case $\tau_1,\tau_2 <\infty$, one could choose $\bar \psi =0$ (due to the fact that no strict convexity is required in Theorem~\ref{thm:characterization_maximizers_unbalanced}). 
First of all, we note that $\psi_n \in C^\infty$, and if $\psi \geq m$ for some $m \in \R_+$, then $\psi_n \geq m$ for every $n \in \N$. Furthermore, $\psi_n$ is also superlinear at $+\infty$, since by Fatou's lemma
\begin{align}
    \liminf_{x \to +\infty} 
        \frac{\psi_n(x)}{x}
            \geq
    \liminf_{x \to +\infty} 
        \int 
            \frac{\psi(x-z)}{x} 
                \dd \rho_n(z)
            \geq 
    \int 
    \liminf_{x \to +\infty} 
        \Big( \frac{\psi(x-z)}{x} \Big)
            \dd \rho_n(z)
            = 
        +\infty 
            .
\end{align}
Finally, observe $\psi_n = \bar \psi_n * \rho_n$, where $\bar \psi_n := \psi + n^{-1} \bar \psi$ is a strictly convex function (being a sum of a convex and a strictly convex function). As a consequence, using that $\rho_n \geq 0$, we see that for every $x,y \in \R$, $\lambda \in [0,1]$ 
\begin{align}
    \psi_n(\lambda x + (1-\lambda) y )
        &=
    \int 
        \bar\psi_n
        \big(
            \lambda (x-z) + (1-\lambda) (y-z)
        \big) 
            \dd \rho_n(z) 
\\
        &\leq 
    \lambda \int 
        \bar\psi_n(x-z) 
            \dd \rho_n(z) 
        +
    (1-\lambda) \int 
        \bar\psi_n(y-z) 
            \dd \rho_n(z) 
\\
        &=
    \lambda \psi_n(x)  
        +
    (1-\lambda) \psi_n(y)    
        \, , 
\end{align}
with equality if and only if, for every $z \in \R$, 
\begin{align}
    \bar \psi_n 
    \big(
        \lambda (x-z) + (1-\lambda) (y-z)
    \big)
        \rho_n(z)
            =
    \Big(
         \lambda \bar\psi_n(x-z) 
              +
        (1-\lambda) \bar\psi_n(y-z)
    \Big) 
        \rho_n(z)
            .
\end{align}
Using that $\rho_n >0$ on $[-\frac1n,\frac1n]$, and that $\bar \psi_n$ is strictly convex, we infer that $x=y$, which precisely shows that $\psi_n$ is strictly convex as well.

We thus constructed a smooth, bounded from below, strictly convex function $\psi_n$, which we claim to be a good approximation of $\psi$, in a suitable uniform sense.
% 
% NEW PROOF
First of all, note that $\psi$ being the Legendre transform of a convex map $\varphi$ satisfying \eqref{eq:assumptions_varphi_intro} implies that $\psi$ is monotone non-decreasing as well as 
\begin{align}
\label{eq:Lip_Linf_intervals_psi}
   \lim_{x \to -\infty} 
        \psi(x) 
            =
    \inf_{x \in \R} 
        \psi(x)
            = 
        -\varphi(0) \in \R 
            ,
\end{align}
see Remark~\ref{rem:monotonicity_Legendre} in the appendix for a proof. In particular, this implies that 
\begin{align}
\label{eq:psi_loc_above_bounded}
    \psi \in L^\infty(I) \cap \text{Lip}(I)
        \, , \quad 
    \text{for every interval } I = (-\infty, b]
        \, , \, b \in \R 
            .
\end{align}
The same holds by construction for $\bar \psi$. 
As a consequence, we deduce that, for every $b \in \R$, 
\begin{align}
    \sup_{x \in (-\infty,b]} 
        | \psi_n(x) - \psi(x)|
        &\leq 
    \sup_{x \in (-\infty,b]}
        | \bar \psi_n * \rho_n(x) - \psi * \rho_n(x) |
            +
        | \psi * \rho_n(x) - \psi(x) |
\\
        &= 
    \sup_{x \in (-\infty,b]}
        \frac1n 
        | \bar \psi * \rho_n(x) |  
            +
        |  \psi * \rho_n(x) - \psi(x) |
\\  
        &\leq 
    \frac1n 
    \sup_{x \in (-\infty,b+1]}
        |\bar \psi(x)| 
            +
        \text{Lip}(\psi)( (-\infty,b+1]) 
            \int |y| \dd \rho_n(y) 
\\
        &\leq 
    \frac1n 
    \Big(
        \sup_{x \in (-\infty,b+1]}
        |\bar \psi(x)| 
            +
        \text{Lip}(\psi)( (-\infty,b+1]) 
    \Big)
                .
\end{align}
which clearly goes to $0$ as $n \to \infty$.
This in particular shows that $\psi_n \to \psi$ in $L^\infty_{\text{loc}}(\R)$.
% \footnote{At this level of generality, this can only hold locally, due to the fact that $\psi$ is necessarily unbounded.
% I have the feeling that by choosing $\rho_n$ more attentively, we could achieve a global $L^\infty$ bound on the difference, but I don't think it is necessary.} 
% as $n \to \infty$. 

Note due to the fact that $\rho$ (hence $\rho_n$) is symmetric (and thus in particular $\int z \dd \rho_n(z) = 0$) and $\bar \psi \geq 0$, we have that 
\begin{align}
    \psi_n (x) 
        \geq 
    \psi * \rho_n (x) 
        \stackrel{\text{Jensen}}{\geq}
    \psi 
    \Big(
        \int (x-z) \dd \rho_n(z)
    \Big)
        = 
    \psi(x) 
        \, , \qquad \forall x \in \R \, .
\end{align}

We consider the quantum optimal transport problems associated with $\psi_n$ and $\varphi_n := \psi_n^*$, and denote by $\dualufn $, $\primalun$ respectively the dual and the primal functional associated to $\psi_n$, and with $\dualun $, $\primtaun$ the dual and the primal problem, respectively.

By definition of dual functional, together with \eqref{eq:dualf>dualfu}, this shows that $\dualufn  \leq \dualfu  \leq \dualf $, for every $n \in \N$. In particular, it is straightforward to check that the validity of Proposition~\ref{prop:superlevel_bdd} for  $\dualf $ implies the following equi-coercivity property (up to translation) for $\dualufn $.

\begin{proposition}[Equi-coercivity]
\label{prop:superlevel_bdd_approx}
%Let $\cH_1,~\cH_2$ be Hilbert spaces of dimensions $d_1$ and $d_2$, respectively. 
% Let $\rho\in H_\geq(\cH_1)$ and $\sigma\in H_\geq(\cH_2)$ be definite positive density matrices on $H_1$ and $H_2$, respectively, and let $\psi:\R\to\R$ be a superlinear convex function bounded from below by a constant $m\in\R$. 
For every  $U\in \rmH(\cH_1)$, $V\in \rmH(\cH_2)$ such that 
\begin{align} 
        % \tand  
    \emph{D}^{\eps,\tau,n} (U,V)\geq M >-\infty 
        \quad 
    \big(
        \text{and }
        \lambda_1(V)=0 \quad \text{if }\tau = +\infty 
    \big) 
        \,
\end{align}
for some $M\in\R$, there exists a constant $0\leq \mathcal{A}=\mathcal{A}(M) <\infty$ independent of $n$ such that $||U||_\infty,||V||_\infty\leq \cA$.  In particular, whenever a sequence $(U_n,V_n) \in \rmH(\cH_1) \times \rmH(\cH_2)$ satisfies, for every $n\in \N$,
\begin{align}
    \inf_{n \in \N}
        \emph{D}^{\eps,\tau,n} (U_n,V_n)\geq M >-\infty 
            \quad 
        \big(
            \text{and }
                \lambda_1(V_n)=0 
         \quad \text{if }\tau = +\infty   
        \big)
            ,
\end{align}
then $\{ (U_n,V_n)\}_n$ is bounded, and therefore pre-compact.
\end{proposition}

We are ready to prove Theorem~\ref{thm:duality_general}. 

\begin{proof}[Proof of Theorem~\ref{thm:duality_general}]
We apply  Theorem~\ref{thm:characterization_maximizers}, Theorem~\ref{thm:characterization_maximizers_unbalanced} to $\psi_n$ and deduce that
\begin{align}
\label{eq:duality_approx}
    \dualun  = \primtaun
        \, , \qquad \forall n \in \N 
            .
\end{align}
We claim we can pass to the limit the previous equality as $n \to \infty $ to conclude the duality result for $\psi$. We start by observing that for every $(U_n,V_n) \to (U, V) \in \rmH(\cH_1) \times \rmH(\cH_2)$, we have that 
\begin{align}
\label{eq:gamma_approx_dual}
    \lim_{n \to \infty} 
        \dualufn (U_n,V_n) 
            = 
    \dualfu (U,V)
        ,
\end{align}
which is a stronger statement than $\Gamma$-convergence (every approximating sequence is a recovery sequence).
This follows from the fact that converging sequences in $\rmH(\cH_1) \times \rmH(\cH_2)$ are bounded and that $\psi_n \to \psi$ in $L^\infty_{\text{loc}}(\R)$. 
% \lore{(Exercise! [EC: ma vaff.[NM: did it]] :P)}. 

% \lore{HERE}

This in particular shows that $-\dualufn  \stackrel{\Gamma}{\to} -\dualfu $. Together with \eqref{eq:duar_renormalized} and Proposition~\ref{prop:superlevel_bdd_approx}, by the fundamental theorem of $\Gamma$-convergence, it provides the claimed convergence at the level of the dual problems
\begin{align}
    \lim_{n \to \infty}
        \dualun 
    =  \dualu 
        .
\end{align}

We are left to show the convergence of the primal problems, which follows from very similar arguments, namely by showing that $\primalun \stackrel{\Gamma}{\to} \primaltau$ plus suitable compactness. 

% 
% NEW PROOF
In order to show the first statement, we first claim that $\psi_n \to \psi$ in  $L^\infty_{\text{loc}}(I)$ for $I$ bounded from above (cfr. \eqref{eq:psi_loc_above_bounded}) implies $\varphi_n \to \varphi$ in  $L^\infty_{\text{loc}}(\R_+)$. To prove this, we start observing that for every compact $K \subset \R_+$, $\varphi$ is continuous, and thus bounded, on $K$. Therefore, for every $p \in K$, we have that 
\begin{align}
    - 
    \| \varphi \|_{L^\infty(K)}
        \leq 
    \varphi(p) 
        =        
    \sup_{x \in \R} 
    \left\{
        p x - \psi(x) 
    \right\} 
        % \leq 
    % \| \varphi \|_{L^\infty(K)}
        .
\end{align}
This in particular shows that the sup in the previous equation can be taken over the set 
\begin{align}
    C(p)
        := 
    \Big\{
        x \in \R 
            \suchthat 
    %     -2
    % \| \varphi \|_{L^\infty(K)} 
    %     \leq 
    p x - \psi(x)
        \geq 
    - \| \varphi \|_{L^\infty(K)} -1 
    \Big\}
        .
\end{align}
% Denote by $C(p)$ the same set with $\psi$ instead of $\psi_n$.
% Using that $\psi_n \geq \psi$, we directly infer that $C_n(p) \subset C(p)$, 
% $\subset C(p_K)\cup C(0)$ 
% for every $n \in \N$ and for every $p \in K$.
As a  consequence of the Fenchel--Young inequality, we infer that, for $p \in K$, 
\begin{align}
    \varphi(p) 
        = 
    \sup_{x \in C(p)} 
    \left\{
        p x - \psi(x) 
    \right\}
        &\leq 
    \sup_{x \in C(p)}
    \left\{
        \varphi_n(p) + \psi_n(x) - \psi(x) 
    \right\}
    \\
        &\leq 
    \varphi_n(p) + \| \psi - \psi_n \|_{L^\infty(C(p))}
        .
\end{align}
On the other hand, we know that $\varphi_n \leq \varphi$ everywhere (as $\psi \leq \psi_n$), hence we conclude that 
\begin{align}
\label{eq:bastaa}
    \| \varphi - \varphi_n \|_{L^\infty(K)}
        \leq 
    \sup_{p \in K}
    \| \psi - \psi_n \|_{L^\infty(C(p))}
        \, , \quad 
    \forall n \in \N 
        .
\end{align}
Let $p_K$ be the max of $K$. Then, by definition of the sets themselves, we have that for every $p \in K$, 
\begin{align}
    C(p) \cap \R_+ \subset C(p_K)
        \tand 
    C(p_K) \subset (-\infty , \alpha_K] 
        , 
\end{align}
for some $\alpha_K < +\infty$, by superlinearity of $\psi$, see Remark~\ref{rem:superlinear_prop1}. This, together with \eqref{eq:bastaa}, provides the estimate
\begin{align}
    \| \varphi - \varphi_n \|_{L^\infty(K)}
        \leq 
    \| \psi - \psi_n \|_{L^\infty ((-\infty,\alpha_K] )}
        \, , \quad 
    \forall n \in \N 
        .
\end{align}

Recalling that from \eqref{eq:psi_loc_above_bounded} we know that $\psi_n \to \psi$ uniformly on every set which is bounded from above, we conclude that $\varphi_n \to \varphi$ in  $L^\infty_{\text{loc}}(\R_+)$.

As in \eqref{eq:gamma_approx_dual}, we now take advantage of this fact  to ensure that
\begin{align}
\label{eq:gamma_approx_primal}
    \lim_{n \to \infty} 
        \primalun(\Gamma_n) 
            = 
    \primaltau(\Gamma)
        \, , \quad \forall \,  
    \Gamma_n \to \Gamma \in \DM(\cH_1 \otimes \cH_2)
        .
\end{align}
% \lore{Add one line saying that it works exactly as in the dual case, where we may out few additional details -- NATA notes.}
% Note that the fact that $\varphi_n \to \varphi$ in  $L^\infty_{\text{loc}}(\R_+)$ can be proved directly using that $\psi_n \geq \psi$ and that $\psi_n \to \psi$ in  $L^\infty_{\text{loc}}(\R)$. 
% BELOW: true but useless. Do not give any coercivity unfortunately.
% First of all, we observe that $\psi_n \geq \psi$ implies that $\varphi_n = \psi_n^* \leq \psi^* = \varphi$, and therefore $\primalun \leq \primal$ for every $n \in \N$. This ensures the validity of the limsup inequality using the constant sequence.
Secondly, let $(U_n^*,V_n^*)$ be the maximizer of $\dualufn $ (with $\lambda_1(V_n^*)=0$ if $\tau = +\infty$). By Theorem~\ref{thm:characterization_maximizers}, Theorem~\ref{thm:characterization_maximizers_unbalanced}  we know that 
\begin{align}
    \Gamma_n^* 
        := 
    \psi_n' \left[ \frac{U_n^* \oplus V_n^* - C}{\varepsilon} \right]    
    %     =
    % (\psi_n')_+ \left[ \frac{U_n^* \oplus V_n^* - C}{\varepsilon} \right]   
\end{align}
is the unique minimizer of $\primalun$.
% Here $(\psi_n')_+$ denotes the right derivative of $\psi_n$, which of course it coincides with $\psi_n'$ everywhere, due to $\psi_n$ being $C^1$. 
Thanks to Proposition~\ref{prop:superlevel_bdd_approx}, we know that, up to a non-relabeled subsequence, $(U_n^*,V_n^*) \to (U^*,V^*)$ as $n \to \infty$. Moreover, with $\psi$ being convex, we have that $\psi$ is locally Lipschitz. In particular, by the properties of the convolution this implies that, for every bounded set $I \subset \R$, one has that 
\begin{align}
    \sup_{n \in \N} 
        \| \psi_n' \|_{L^\infty(I)} 
    \leq
        \text{Lip} (\psi;B_1(I)) + \text{Lip} (\bar \psi;B_1(I)) 
            < \infty 
        .
    % \psi_n' = (\psi_n')^+ \to (\psi')^+ 
    %     \qquad 
    % \text{locally in } 
    %     L^\infty(\R)
    % \text{ as } n \to \infty 
    %         .
\end{align}
This, together with the boundedness of $(U_n^*,V_n^*)$, implies that $\{ \Gamma_n^* \}_n$ is also bounded, and therefore pre-compact. 
Thanks to the $\Gamma$-convergence showed above in \eqref{eq:gamma_approx_primal}, this in particular shows that every limit point (which exists) must necessarily be a minimizer (the marginal constraints are closed by uniform convergence), thus concluding the proof of Theorem~\ref{thm:duality_general} passing to the limit in \eqref{eq:duality_approx}.
\end{proof}

% \begin{rem}[Possible alternative proof]
% FOR NOW: commented, then we decide.
% If one showed that the corresponding primal functionals $\primalun$ are equi-coercive, then together with the $\Gamma$-convergence shown in \eqref{eq:gamma_approx_primal} one would get the convergence of the primal problems. Starting from this, and using the trivial lower bound for the primal in terms of the dual (weak-duality), one could prove directly the validity of Theorem~\ref{thm:duality_general}.
% % \lore{This is what Nata was talking about, we could also add more details if we want}
% Unfortunately, the equi-coercivity of the primal functionals is less straightforward than the one of the duals, cause $\psi_n \geq \psi$ implies $\varphi_n \leq \varphi$, which shows that the sublevels of $\primaltau$ are included in the one of $\primalun$ and not viceversa. 
% We believe this is true, but we did not find a simple way to tackle it, which is why we instead focused on the dual in the proof above. 
% % But because $\varphi_n \to \varphi$ in $L_{\text{loc}}^\infty$, I think it should be provable by other means (or perhaps using the explicit formulas for $\psi_n$ to infer properties of $\varphi_n$). Because I am not sure, I argued as above directly at the level of the optimizers, which suffices anyhow.
% \end{rem}

\section{Convergence results: from unbalanced to balanced quantum optimal transport}
\label{sec:convergence_resuts}
In this section, we discuss limit behaviors for the unbalanced optimal transport as the parameters $\tau_1$, $\tau_2 \to +\infty$. Recall that we set $\tau=(\tau_1,\tau_2)$, and use the slight abuse of notation $\tau \to +\infty$. We discuss convergence results both at the level of primal and dual problems, as well as the convergence of minimizers/maximizers. 
% 
% Below: not sure which one should come first, maybe idk the convergence of the C-transforms helps for the convergence of the primal/dual.
\begin{theo}[Convergence as $\tau \to +\infty$]
\label{thm:convergence_unb_to_b}
    Under the standing assumption of Section~\ref{sec:unbalanced_noncommutative}, we have the following convergence results:
    \begin{enumerate}
        \item The functionals $\primaltau$  $\Gamma$-converge as $\tau \to \infty$ to the functional 
        \begin{align}
            \Gamma \mapsto 
            \begin{cases}
                \primal(\Gamma)
                    &\text{if } \Gamma \mapsto (\rho,\sigma) 
            \\
                +\infty 
                    &\text{otherwise} 
                        \, .
            \end{cases} 
            %     \qquad 
            % \big( \text{resp. } - \emph{D}^\eps  \big)
                % .
        \end{align}
        % $\primal$ with marginals constraint $\Gamma \mapsto (\rho,\sigma)$ (resp. $\emph{D}^\eps $).
        \item The unique minimizer  $\Gamma^{*,\tau} \in \rmH(\cH_1 \otimes \cH_2)$ of $\primaltau$ converges (up to subsequence) as $\tau \to +\infty$ to a minimizer $\Gamma^* \in \rmH(\cH_1 \otimes \cH_2)$ of $\primal$, under the constraint $\Gamma^*  \mapsto (\rho,\sigma)$. 
        \item We have the following convergence result for the dual:
        for every $(U_\tau,V_\tau) \to (U,V) \in \rmH(\cH_1) \times \rmH(\cH_2)$ as $\tau \to +\infty$ one has that 
        \begin{align}
        \label{eq:dual_strong_convergence}
            \lim_{\tau \to +\infty}
               \emph{D}^{\eps,\tau} (U_\tau,V_\tau) = \emph{D}^\eps (U,V)
                    .
        \end{align}
        In particular, $-\emph{D}^{\eps,\tau} $ $\Gamma$-converges to $-\emph{D}^\eps $.
        \item The unbalanced optimal transport problems converge as $\tau \to +\infty$  to the balanced optimal transport problem, i.e.
        \begin{align}
        \label{eq:conv_duality}
            \lim_{\tau_1,\tau_2 \to +\infty} 
                \primtau
                    =
                \prim
            %         \quad  \tand  \quad  
                    =
                \dual 
                    =
            \lim_{\tau_1,\tau_2 \to +\infty} 
                \dualu 
                    ,
        \end{align}
        and the minimizers $\Gamma^{*,\tau}$ converge (up to subsequence) to a minimizer of $\primal$ with marginal constraints.
        \item Let $(U^{*,\tau}, V^{*,\tau}) \in \rmH(\cH_1) \times \rmH(\cH_2)$ be the unique maximizer of $\emph{D}^{\eps,\tau} $. Define the recentered potentials 
        \begin{align}
            \big(
                \hat U^{*,\tau}, \hat V^{*,\tau}
            \big) 
                := 
            \big( 
                U^{*,\tau} + \lambda_1(V^{*,\tau}) {\rm Id}_{\cH_1}
                    , 
                V^{*,\tau} - \lambda_1(V^{*,\tau})
            {\rm Id}_{\cH_2}\big)
                \in \rmH(\cH_1) \times \rmH(\cH_2)
                    .
        \end{align}
        Then $(\hat U^{*,\tau}, \hat V^{*,\tau})$ converges (up to subsequence)  as $\tau \to +\infty$ to a maximizer $(U^*,V^*) \in \rmH(\cH_1) \times \rmH(\cH_2)$ of $\emph{D}^\eps $ which satisfies $\lambda_1(V^*)=0$.
    \end{enumerate}
    If we assume additionally that $\psi \in C^1$ and strictly convex, then the $\Gamma^{*,\tau}$ converge to the unique minimizer in $\prim$, while $(\hat U^{*,\tau}, \hat V^{*,\tau})$ converges to the unique maximizer $(U^*,V^*) \in \rmH(\cH_1) \times \rmH(\cH_2)$ of $\emph{D}^\eps $ which satisfies $\lambda_1(V^*)=0$.
\end{theo}
\begin{rem}[Lack of coercivity in the dual]\label{rmk:dualcoercivity}
Note that $\dualf $ is not coercive, since the sequence $\{(a{\rm Id},-a{\rm Id})\}_{a \in \mathbb{R}}$ is such that $\dualf (a{\rm Id},-a{\rm Id}) =C_0 \in \R$, but it does not admit a convergent subsequence. By Theorem~\ref{thm:convergence_unb_to_b}, and by the fundamental theorem of $\Gamma$-convergence, this shows that the family of functionals $\dualfu$ is not equi-coercive (note indeed that the bounds obtained in the proof of Proposition~\ref{thm:characterization_maximizers_unbalanced} are not uniform in $\tau_1,\tau_2$). This explains why it is not directly possible to infer the main results for the dual problem $\dual$ directly from the unbalanced setting by sending the parameter to infinity. For a similar reason, the $\Gamma$-convergence result proved in 1. in the above theorem does not directly provide the validity of 3., which instead requires a different argument.
\fr
\end{rem}

%\lore{In 4., it is not clear if we really need to recenter the maximizers. It would be interesting to try to understand if we really need to do that, but unclear.}

% \begin{rem}[Nonsmooth $\psi$]
%     The convergence result \eqref{eq:conv_duality} in Theorem~\ref{thm:convergence_unb_to_b} holds in fact in a more general setting, namely even if $\psi$ is not $C^1$ nor strictly convex, but rather only superlinear, convex, and bounded from below. This is due to the fact that duality still holds in this setting, see in particular Theorem~\ref{thm:duality_general}.
% \end{rem}

% \lore{Part of this proof is essentially the same as Nataliia and Augusto did in the other paper, we could think of omit that. In particular the $\Gamma$-convergence of the primal.}

\begin{proof}
% \nat{I believe this sentence is useless now since you changed the notation already} For simplicity of notation, we use the short-hand notation $\tau = (\tau_1,\tau_2)$, and the same for the functionals $\text{F}^{\eps,\tau}  = \primaltau$, $ \text{D}^{\eps,\tau}  = \dualfu $. 
% 
% \smallskip
% \noindent 
1. \ 
First, we show the claimed $\Gamma$-convergence for the primal functionals. We prove the $\liminf$ and the $\limsup$ inequality, starting from the first one. Let $\Gamma_\tau \to \Gamma \in \rmH(\cH_1 \otimes \cH_2)$ as $\tau \to +\infty$, and assume without loss of generality that $\sup_\tau \text{F}^{\eps,\tau} (\Gamma_\tau) <\infty$. By the continuity of $\primal$, we know that 
\begin{align}
\label{eq:bound_conv_0}
    \lim_{\tau \to +\infty} 
        \primal(\Gamma_\tau) 
            =
        \primal(\Gamma) < \infty 
            \, , 
\end{align}
and therefore we also infer that 
\begin{align}
\label{eq:bound_conv_1}
    \sup_{\tau_1, \tau_2 > 0} 
        \tau_1 \cE({\rm P}_1(\Gamma_\tau)| \rho)+ \tau_2 \cE({\rm P}_2(\Gamma_\tau)| \sigma) 
            < \infty 
        .
\end{align}
Note that the partial trace is a continuous operation with respect to the Hilbert-Schmidt convergence of operators. Therefore we have that ${\rm P}_i(\Gamma_\tau) \to {\rm P}_i(\Gamma)$ as $\tau \to \infty$. In  particular, we deduce that 
\begin{align}
\label{eq:bound_conv_2}
    \sup_{\tau_1, \tau_2 > 0} 
        \cE({\rm P}_1(\Gamma_\tau)| \rho) +  \cE({\rm P}_2(\Gamma_\tau)| \sigma) 
            < \infty 
        .
\end{align}
Combining the estimates in \eqref{eq:bound_conv_1} and \eqref{eq:bound_conv_2}, from the convergence of the marginals and the continuity of the relative entropy we conclude that 
\begin{align}
    \cE({\rm P}_1(\Gamma)| \rho)
        +
    \cE({\rm P}_2(\Gamma)| \sigma)
        =
    \lim_{\tau \to + \infty}
        \cE({\rm P}_1(\Gamma_\tau)| \rho)
            +
        \cE({\rm P}_2(\Gamma_\tau )| \sigma)
        =
    0   
        .
\end{align}
By using that the relative entropy is a nonnegative functional, this shows that both relative entropies must vanish, which in particular yields $\Gamma \mapsto (\rho,\sigma)$. 
Combining this with \eqref{eq:bound_conv_0}, we end up with
\begin{align}
    \liminf_{\tau \to +\infty}
        \text{F}^{\eps,\tau} (\Gamma_\tau)
    \geq 
        \primal(\Gamma) 
        \quad \tand \quad 
    \Gamma \mapsto (\rho,\sigma)
        ,
\end{align}
which concludes the proof of the $\liminf$ inequality.
For the $\limsup$ inequality, it is enough to choose the constant recovery sequence. 
% Let us now show the $\Gamma$-convergence of the dual functionals. The $\limsup$ inequality follows from the comparison \eqref{eq:dualf>dualfu} between balanced and unbalanced dual functionals, 
% % and the continuity of $\dualf $
% by choosing the constant recovery sequence \nat{so we are in the end showing for $\text{D}^{\eps,\tau} $, not $-\text{D}^{\eps,\tau} $?}.  
% \begin{align}
%     \liminf_{\tau \to +\infty}
%         \text{D}^{\eps,\tau} (U_\tau,V_\tau)
%         \geq
%     \liminf_{\tau \to +\infty}
%         \dualf (U_\tau,V_\tau)
%         =
%     \dualf (U,V)
%         , 
% \end{align}

\smallskip 
\noindent 
2. \ 
Note that $\{ \text{F}^{\eps,\tau}  \}_\tau$ is equi-coercive, due to the fact that $\text{F}^{\eps,\tau}  \geq \primal$, for every $\tau>0$, together with the fact that $\primal$ has bounded sublevels.
As a consequence, 2. follows directly from 1. by the fundamental theorem of $\Gamma$-convergence.

\smallskip
\noindent 
3. \ 
Let us now show the convergence claimed in \eqref{eq:dual_strong_convergence}. 
To this purpose, we observe that 
\begin{align}
    \lim_{\tau_1 \to +\infty}
        -\tau_1 \Tr
        \left[ 
            e^{-\frac{U_{\tau_1}}{\tau_1} + \log \rho} - \rho 
        \right] 
        =
    \lim_{t \to 0^-}
        \frac
            {
            \Tr
            \left[ 
                e^{\log \rho + t \tilde U_t} - e^{\log \rho} 
            \right] 
            }
            {t}
        =  
    \frac{\dd }{ \dd t }\bigg|_{t=0}
        g(\log \rho + t U) 
            ,
\end{align}
where we set $\tilde U_t := U_{-1/t}$ and $g(A):= \Tr[e^A]$. The last equality is a consequence of the following observation: let $B\subset \tH(\cH_1)$ be the ball of radius $1$ centered in $\log \rho \in \tH(\cH_1)$. For $t$ sufficiently small, both $\log\rho + t \tilde U_t$ and $\log \rho + t U$ belong to $B$, hence by \eqref{eq:continuity_tracevarphi} we can estimate
\begin{align*}
     \bigg|
     \frac1t
     \Tr
    \left[ 
        e^{\log \rho + t \tilde U_t} - e^{\log \rho +  t U} 
    \right]
    \bigg|
        \leq 
    \text{Lip} (g ; B)
    \|
         (\tilde U_t -U)
    \|_{\text{HS}}
        \to 0
            , 
\end{align*}
as $t \to 0^-$.  The claimed equality readily follows.
% 
% OLD EST WITH KLEIN -- now equality
% 
% \begin{align}
%     \liminf_{\tau_1 \to +\infty}
%         -\tau_1 \Tr
%         \left[ 
%             e^{-\frac{U_{\tau_1}}{\tau_1} + \log \rho} - \rho 
%         \right] 
%         =
%     \liminf_{t \to 0^-}
%         \frac
%             {
%             \Tr
%             \left[ 
%                 e^{\log \rho + t \tilde U_t} - e^{\log \rho} 
%             \right] 
%             }
%             {t}
%         \geq  
%     \frac{\dd }{ \dd t }\bigg|_{t=0}
%         g(\log \rho + t U) 
%             ,
% \end{align}
% where we set $\tilde U_t := U_{-1/\tau_1}$ and $g(A):= \Tr[e^A]$. The inequality above can, for example, be proven once again by Klein's inequality, which in particular implies 
% \begin{align*}
%      \frac1t\Tr
%     \left[ 
%         e^{\log \rho + t \tilde U_t} - e^{\log \rho +  t U} 
%     \right] 
%         \geq 
%     \Tr 
%     \left[
%         e^{\log \rho + t \tilde U_t}
%          (\tilde U_t -U)
%     \right]
%         \to 0
% \end{align*}
% as $t \to 0^{-}$, where we used that $\tilde U \to U$. \lore{I honestly think it is equality, but because the $\geq$ is easier thx to Klein, for now let's leave it like this.}
% 
An application of Lemma~\ref{lem:gradient} yields  $\nabla g(A) = e^A$ as well as that the $\liminf$ is in fact a limit, and it holds
\begin{align}
    \lim_{\tau_1 \to +\infty}
        -\tau_1 \Tr
        \left[ 
            e^{-\frac{U}{\tau_1} + \log \rho} - \rho 
        \right] 
        =
   \langle 
        (\nabla g)(\log \rho)
            , 
        U
   \rangle
        =
    \Tr [ \rho U ]
        .
\end{align}

% OLD NOT NEEDED ANYMORE
% \lore{Also: it might be doable to use Klein's directly without talking about derivatives, if one justifies this computation (in particular that limit): for $ t < 0$
% \begin{align*}
%      \frac1t\Tr
%     \left[ 
%         e^{\log \rho + t \tilde U_t} - e^{\log \rho } 
%     \right] 
%         \geq 
%     \Tr 
%     \left[
%         e^{\log \rho + t \tilde U_t}
%          \tilde U_t 
%     \right]
%     %     =
%     % \Tr 
%     % \left[
%     %     e^{\log \rho + t \tilde U_t}
%     %      (\tilde U_t - U) 
%     % \right]   
%     %     +
%     % \left[
%     %     e^{\log \rho + t \tilde U_t}
%     %     U
%     % \right]    
%         \to 
%     \Tr[e^{\log \rho} U ] 
%         = 
%     \Tr[\rho U] 
% \end{align*}
% }

The very same argument applied on $\cH_2$ ensures that 
\begin{align}
    \lim_{\tau_2 \to +\infty}
        -\tau_2 \Tr
        \left[ 
            e^{-\frac{V}{\tau_2} + \log \sigma} - \sigma 
        \right] 
        =
    \Tr [ \sigma V ]
        .
\end{align}
The latter two equalities prove exactly that 
$  \displaystyle 
\lim_{\tau \to +\infty} \text{D}^{\eps,\tau} (U_\tau,V_\tau) = \dualf (U,V)
$, hence providing the claimed convergence.

\smallskip 
\noindent 
4.  \ 
The convergence of the primal problems follows from the same argument above in the proof of 2. The convergence (and equality) of the dual problems follows by duality, in particular Theorem~\ref{thm:duality_general}.

\smallskip 
\noindent 
5. \ 
We take advantage of Remark~\ref{rem:symmetry} and once again of \eqref{eq:dualf>dualfu}, to deduce that
\begin{align}
\label{eq:lb_final}
    \dualf 
    \big( 
        \hat U^{*,\tau}, \hat V^{*,\tau}
    \big)
        =
    \dualf 
    \big( 
        U^{*,\tau},  V^{*,\tau}
    \big)
        \geq  
    \text{D}^{\eps,\tau} 
    \big( 
        U^{*,\tau},  V^{*,\tau}
    \big)
        = 
    \mathfrak{D}^{\eps,\tau} 
        =
    \mathfrak{F}^{\eps,\tau}
        ,
\end{align}
where in the last step we used the duality result proved in Theorem~\ref{thm:duality_general}. We use the convergence proved in 3.\ to infer
\begin{align}
    \liminf_{\tau \to +\infty}
        \dualf 
        \big( 
            \hat U^{*,\tau}, \hat V^{*,\tau}
        \big)
            \geq 
    \liminf_{\tau \to +\infty}
        \mathfrak{F}^{\eps,\tau}   
            =
   \prim
            > 
    - \infty 
        .
\end{align}
This implies that the sequence 
$ \{  
\big( 
    \hat U^{*,\tau}, \hat V^{*,\tau}
\big) 
\}_\tau  $
belongs to a superlevel set of $\dualf $, and thus by construction satisfies the hypothesis of Proposition~\ref{prop:superlevel_bdd}. It follows that it is a bounded sequence of operators, and therefore it is pre-compact. We conclude the proof of 4.\ by showing that any limit point is necessarily a maximizer $(U^*,V^*)$ of $\dualf$ with $\lambda_1(V^*)=0$. Assume indeed that 
$\big( 
    \hat U^{*,\tau}, \hat V^{*,\tau}
\big) \to (U^*,V^*)$ as $\tau \to +\infty$. The fact that $\lambda_1(V^{*,\tau})=0$ for every $\tau>0$ ensures that $\lambda_1(V^*)=0$. Finally, by \eqref{eq:lb_final} together with the continuity of $\dualf $ we deduce that
\begin{align}
    \dualf (U^*,V^*)
        =
    \lim_{\tau \to +\infty} 
        \dualf 
        \big( 
            \hat U^{*,\tau}, \hat V^{*,\tau}
        \big)   
            \geq 
    \prim
        = 
    \dual 
        ,
\end{align}
where at last we used the duality result Theorem~\ref{thm:duality_general}. We thus conclude that $(U^*,V^*)$ is indeed a maximizer with $\lambda_1(V^*)=0$.
\end{proof}

% \subsubsection{$\Gamma$-convergence of the primal/dual problems}
% [...]

\subsection{Convergence of \texorpdfstring{$(C,\psi,\ep,
\tau)-$transforms}{}}
% \nat{need to state that $\psi$ is strictly convex $C^1$ again} YES
Assume that $\psi$ is strictly convex and $C^1$.
In this short section, we take advantage of the convergence results proved in Theorem~\ref{thm:convergence_unb_to_b} to show that also the associated $(C,\psi,\eps,\tau)$-transforms converge, as the parameters $\tau \to +\infty$. More precisely, we claim that 
\begin{align}
    \mathscr{F}_1^{(C,\psi,\eps,\tau)}(V_\tau)
        \to 
    \mathscr{F}_1^{(C,\psi,\eps)}(V)
        \qquad 
    \forall \, \, 
        V_\tau \to V 
            .
\end{align}
Recall from the previous section that the following property (stronger than $\Gamma$-convergence) holds: one has that \eqref{eq:dual_strong_convergence}
\begin{align}
    \lim_{\tau \to +\infty}     
        \text{D}^{\eps,\tau} (U_\tau,V_\tau) 
    = 
        \dualf (U,V)
            \, , \qquad 
         \text{for every} \quad  (U_\tau, V_\tau) \to (U,V)
            \, .
\end{align}
In particular, we have the $\Gamma$-convergence of the one-parameter functionals, given by 
\begin{align}
    - \text{D}^{\eps,\tau} (\cdot, V_\tau) 
        \xrightarrow[\tau\to+\infty]{\Gamma}
    -\text{D}^\eps (\cdot,V)
        , 
\end{align}
for every $V_\tau \to V$. Using once again the inequality \eqref{eq:dualf>dualfu} together with the coercivity of the dual functional $\dualf$ (cfr. Proposition~\ref{prop:superlevel_bdd}, see also Remark~\ref{rem:shift}), it readily follows that the sequence of functionals $\{- \text{D}^{\eps,\tau} (\cdot, V_\tau) \}_\tau$ is equi-coercive. This, together with the aforementioned $\Gamma$-convergence, provides the convergence of the associated minimizers, which precisely means 
\begin{align}
    \lim_{\tau \to +\infty} 
        \mathscr{F}_1^{(C,\psi,\eps,\tau)}(V_\tau)
    = 
        \mathscr{F}_1^{(C,\psi,\eps)}(V)
            \qquad \forall \,  V_\tau \to V 
                , 
\end{align}
as claimed.

% [...] TODO
% 
% \lore{
% For every $V_\tau \to V_\infty$ as $\tau \to \infty$, we have that
% \begin{align}
%     \dualf 
%     \big(
%         \mathscr{F}_1^{(C,\psi,\eps,\tau)}(V_\tau), V_\tau
%     \big)
%         \stackrel{\eqref{eq:dualf>dualfu}}{\geq} 
%     \dualfu 
%     \big(
%         \mathscr{F}_1^{(C,\psi,\eps,\tau)}(V_\tau), V_\tau
%     \big)
%         \geq    
%     \dualfu 
%     \big(
%         U, V_\tau
%     \big) 
%         \stackrel{\Gamma-\text{conv}}{\geq}
%     \dualf 
%     \big(
%         U, V_\infty
%     \big) 
% \end{align}
% for every $U \in \rmH(\cH_1)$. Therefore if we can show that $\{  \mathscr{F}_1^{(C,\psi,\eps,\tau)}(V_\tau) \}_\tau$ is bounded, then we would conclude (by the continuity of $\dualf $) that 
% \begin{align}
%     \lim_{\tau \to +\infty} 
%         \mathscr{F}_1^{(C,\psi,\eps,\tau)}(V_\tau)
%     = 
%         \mathscr{F}_1^{(C,\psi,\eps)}(V_\infty)
%             \qquad \forall \,  V_\tau \to V_\infty 
%                 .
% \end{align}
% This actually is also clear from the optimality conditions (and we should add a remark about this), but this proof does not necessarily requires smoothness of $\psi$ so it might a bit more general.
% 
% 
% We are left to show that $\{  \mathscr{F}_1^{(C,\psi,\eps,\tau)}(V_\tau) \}_\tau$ is bounded. [TODO]
% }
% 

\appendix
\section{Legendre's transform in a non-commutative setting}
\label{appendix:primal_dual}
In this appendix, we study the properties and formula for the Legendre transform of functions obtained by composing the trace operator with (the lifting of) a convex map on the real line. More specifically and throughout all the section, let $\varphi \colon [0,+\infty) \to [m,+\infty)$ be convex, superlinear at infinity, namely $\lim_{t \to +\infty}\varphi(t)/t =+\infty$, and $m > -\infty$. Recall that, given such $\varphi$, its Legendre transform is given by 
\begin{align}
    \psi: \R \to \R 
        \, , \quad 
    \psi(x) := 
    \sup_{t \in \R} 
    \left\{
        x t - \varphi(t) 
    \right\}
        \, , \quad 
    \forall x \in \R 
        ,
\end{align}
where we extend $\varphi$ to $+\infty$ over the negative half-line. 
% 
    % {\color{red} [lore] Either allow $+\infty$ as an image or allow to have a subset of $\R$ in domain, I prefer the first one -- check that all the rest is still fine}

    % {\color{red} [lore] We probably need only at $t \to +\infty$, is it equivalent for $\varphi$ or $\psi=\varphi^*$?}
%
\begin{rem}[Monotonicity of the Legendre's transform]
\label{rem:monotonicity_Legendre}
    Under the aforementioned assumptions on $\varphi$, we claim that $\psi$ is monotone non-decreasing and satisfies 
    \begin{align}
        \lim_{x \to -\infty} 
            \psi(x) = 
        \inf_{x\in\R} \psi(x) 
        = - \varphi(0) \in \R 
            \, .
    \end{align}
Indeed, any non-constant real-valued convex function $\psi$ on $\R$ is either 
\begin{itemize}
    \item[(i)] monotone non-decreasing
    \item[(ii)] monotone non-increasing
    \item[(iii)] admits a global minimizer $x^*\in\R$.
\end{itemize}
For (i) the conclusion trivially follows, and we claim that both (ii) and (iii) fail under the assumptions on $\varphi$, since (ii)-(iii) must imply that $\lim\limits_{x\to-\infty}\psi(x) = +\infty$ or/and there exist $x_1<x_2$ such that $\psi(x_1)>\psi(x_2)$. We can exploit this fact to obtain that $\varphi(p)<+\infty$ for some $p <0$, contradicting the original assumption on $\varphi$.

Indeed, let $p_1\in\partial \psi(x_1)$, then $\psi(x)\geq \psi(x_1) + p_1(x-x_1)$. Since $x_1$ is not a minimizer by construction, then $p_1\neq 0$, and in fact $p_1$ must be strictly negative, as 
$$\psi(x_2)\geq\psi(x_1) +p_1(x_2-x_1) > \psi(x_2) + p_1(x_2-x_1).$$
However, then we have $\varphi(p_1) = \psi^*(p_1) = p_1 x_1 - \psi(x_1) < \infty$, contradicting that $\varphi(p)=+\infty$ for all $p<0$.
\fr
\end{rem}

\subsection{Properties of operator liftings}
Given a $d$-dimensional Hilbert space $\mathcal{H}$, the map that sends $f \in C(\mathbb{R})$ into ${\rm Lin}(\rmH(\mathcal{H}),\rmH(\mathcal{H}))$ via spectral calculus is a vector space homomorphism and preserves the natural composition and product operations in $C(\mathbb{R})$ and ${\rm Lin}(\rmH(\mathcal{H}),\rmH(\mathcal{H}))$, namely given $f_1,f_2  \in C(\R)$ and $\lambda \in \mathbb{R}$ we have
\begin{equation}
\label{eq:algebra_liftings}
\begin{aligned}
    &(f_1+f_2)(A) = f_1(A) + f_2(A),\qquad &&(\lambda f_1)(A) = \lambda f_1(A),\\
    & (f_1\cdot f_2)(A) = f_1(A) f_2(A), && (f_1 \circ f_2)(A) =f_1(f_2(A)),
\end{aligned}
\end{equation}
for every $A \in \rmH(\cH)$.

It can be proved (see for instance \cite[Thm.\ 2.10]{carlen2010trace}) that, if $\varphi \colon [0,+\infty) \to \mathbb{R}$ is convex (respectively strictly convex), then $\rmH_\geq(\cH)\ni \Gamma \mapsto \Tr[\varphi(\Gamma)] \in \mathbb{R}$ is convex (respectively strictly convex).
In particular, since real-valued convex functions on $\mathbb{R}^n$ are locally Lipschitz we have that 
\begin{equation}
\label{eq:continuity_tracevarphi}
    \rmH_\geq(\cH)\ni \Gamma \mapsto \Tr[\varphi(\Gamma)] \in \mathbb{R}\qquad\text{ is locally Lipschitz.}
\end{equation}
Additionally, using that $\rmH(\cH) \ni W \mapsto \Tr[W\xi]$ is continuous for every $\xi \in \rmH(\cH)$, we infer  
\begin{align}
\label{eq:dualf_continuous}
    \rmH(\cH_1) \times \rmH(\cH_2) \ni(U,V) \mapsto \dualf (U,V) \in \mathbb{R} 
        \quad 
    \text{is (strictly) concave and continuous}
        \quad 
\end{align}
whenever $\psi$ is (strictly) convex. 

Note that $\Gamma \mapsto \varphi(\Gamma)$ is not necessarily convex (e.g. $\varphi(x) = x^3$, see \cite[Example 2.5]{carlen2010trace}).

% Throughout the rest of the paper, with a slight abuse of notation, we denote by $\varphi$ both the real function and the operator-valued $\varphi$.

\begin{rem}[Primal is bigger than dual]
\label{rem:primal_bigger_dual}
    By \eqref{eq:Fenchel-Young} with $W:= (U \oplus V - C)/\eps$, we obtain   
    \begin{equation}
        \langle U \oplus V - C, \Gamma \rangle 
            \le  
        \eps \Tr\left[ \varphi(\Gamma)\right] 
            + 
        \eps \Tr
        \left[ \psi
            \Big(
                \frac{U \oplus V - C}{\eps }
            \Big)
        \right]
            .
    \end{equation}
    Therefore, for every $\Gamma \mapsto (\rho,\sigma)$, we have that 
    \begin{align}
        \primal(\Gamma) 
            &= 
        \Tr[C \Gamma] + \eps \Tr[\varphi(\Gamma)]
            \geq 
         \langle U \oplus V , \Gamma \rangle 
            - 
        \eps \Tr
        \left[ \psi 
            \Big(
                \frac{U \oplus V - C}{\eps }
            \Big)
        \right] 
    \\
            &=
        \Tr[U \rho] + \Tr[V \sigma]  
            - 
        \eps \Tr
        \left[ \psi 
            \Big(
                \frac{U \oplus V - C}{\eps }
            \Big)
        \right] 
            =
        \dualf (U,V)
            , 
    \end{align}
    for every $(U,V) \in \rmH(\cH_1) \times \rmH(\cH_2)$. The same holds true for the unbalanced case, using the fact that the Legendre transform of the relative entropy is the exponential function, cfr. \eqref{eq:def_Legendre_relative_entropy}.
\fr
    % \lore{Add the unbalanced case as well}\fr
\end{rem}

\begin{lem}[Existence of the maximum]
Let $\varphi$ satisfy the standing assumption of the section and let us denote the lifting of $\varphi$ to the space of Hermitian matrices with a slight abuse of notation by $\varphi$. For every $W \in \rmH(\cH)$
    \begin{equation}
        \sup_{\Gamma \in \rmH_\geq(\cH)} \left( \langle W, \Gamma \rangle - \Tr\left[ \varphi(\Gamma)\right] \right)
            = 
        \max_{\Gamma \in \rmH_\geq(\cH)} \left( \langle W, \Gamma \rangle - \Tr\left[ \varphi(\Gamma)\right] \right)
    \end{equation}
\end{lem}
\begin{proof}
Given $M \in \mathbb{R}$, we claim that the superlevel set $\mathscr{S}_M:=\{ \Gamma :\, \langle W,\Gamma \rangle -\Tr\left[ \varphi(\Gamma)\right] \ge M\}$ is bounded. By contradiction, there exists a sequence $\{ \Gamma^n \}_n \subset \mathscr{S}_M$ such that, denoting $\{\lambda_i^n\}_i$ the eigenvalues of $\Gamma^n$ with $\lambda_1^n \le \lambda_2^n \le \dots \le \lambda_d^n$, we have that $\lambda_d^n \to +\infty$. We estimate
\begin{equation}
\begin{aligned}
    \langle W, \Gamma^n \rangle -\Tr\left[ \varphi(\Gamma^n) \right] 
    &\le d\| W \| \lambda_d^n -\sum_{i=1}^d \varphi(\lambda_i^n)= d\| W \| \lambda_d^n - \varphi(\lambda_d^n)-\sum_{i=1}^{d-1} \varphi(\lambda_i^n)
\\
    &\le d\| W \| \lambda_d^n - \varphi(\lambda_d^n)-\min \{ \varphi \}(d-1),
\end{aligned}
\end{equation}
which converges to $-\infty$, using the property that $\varphi$ is superlinear at infinity.
Thus, the claim is proved. Moreover, notice that the map $\Gamma \mapsto \langle W, \Gamma \rangle - \Tr\left[ \varphi(\Gamma)\right]$, is continuous by \eqref{eq:continuity_tracevarphi}, hence the maximum exists.
\end{proof}
\begin{proposition}[Legendre's transform for functional calculus]
\label{prop:computation_legendre_transform}
    Let $\varphi$ satisfy the standing assumption of the section. We have that \[\Psi(W):=\sup_{\Gamma \in \rmH_\geq(\cH)} \left( \langle W, \Gamma \rangle - \Tr\left[ \varphi(\Gamma)\right] \right)\] satisfies $\Psi(W)=\Tr[\psi'(W)]$, where $\psi'$ is the lifting of $\psi=\varphi^*$ to the space of Hermitian matrices. 
\end{proposition}
\begin{proof}
    We claim that
        \begin{equation}
        \label{eq:Fenchel-Young}
            \langle W, \Gamma \rangle \le  \Tr\left[ \varphi(\Gamma)\right] + \Tr\left[ \psi'(W)\right].
        \end{equation}
    To do this, let us write $\Gamma = \sum_{i=1}^d \Gamma_i \ket{\gamma_i}  \bra{\gamma_i}$ and $W = \sum_{i=1}^d W_i \ket{\xi_j}\bra{\xi_j}$ and compute
    \begin{equation}
    \label{eq:computation_transform}
    \begin{aligned}
        \langle W, \Gamma \rangle = \Tr\left[ W \Gamma \right] 
        &= \sum_{i} \Gamma_i \bra{\gamma_i}W\ket{\gamma_i} = \sum_i\sum_j \Gamma_i W_j |\langle\xi_j |\gamma_i\rangle|^2
    \\
        &
        \le \sum_i \sum_j (\varphi(\Gamma_i) + \psi(W_j)) |\langle\xi_j |\gamma_i\rangle|^2
    \\
        &= \sum_i  \varphi(\Gamma_i) \left(\sum_j|\langle\xi_j |\gamma_i\rangle|^2\right) +  \sum_j \psi(W_j) \left(\sum_i |\langle\xi_j |\gamma_i\rangle|^2\right) 
    \\
        &=  \sum_i  \varphi(\Gamma_i) +  \sum_j \psi(W_j)  = \Tr\left[\varphi(\Gamma)\right]+ \Tr\left[ \psi'(W) \right].
    \end{aligned}
    \end{equation}
    % \begin{equation}
    % \label{eq:computation_transform}
    % \begin{aligned}
    %     \langle W, \Gamma \rangle = \Tr\left[ W \Gamma \right] 
    %     &= \Tr\left[ \sum_i \sum_j \gamma_i\lambda_j \Gamma_i \otimes \Gamma_i w_j \otimes w_j \right] 
    % \\
    %     &= \sum_i \sum_j \gamma_i\lambda_j \Tr\left[ \Gamma_i \otimes \Gamma_i w_j \otimes w_j \right]
    %     =\sum_i \sum_j \gamma_i\lambda_j \langle\Gamma_i, w_j\rangle\,\Tr\left[ \Gamma_i \otimes w_j \right]
    % \\
    %     &= \sum_i \sum_j \gamma_i \lambda_j \langle\Gamma_i, w_j\rangle^2 
    %     \le \sum_i \sum_j (\varphi(\gamma_i) + \psi(\lambda_j)) \langle\Gamma_i, w_j\rangle^2
    % \\
    %     &= \sum_i  \varphi(\gamma_i) \left(\sum_j\langle\Gamma_i, w_j\rangle^2\right) +  \sum_j \psi(\lambda_j) \left(\sum_i \langle\Gamma_i, w_j\rangle^2 \right) 
    % \\
    %     &= \sum_i  \varphi(\gamma_i) \|\Gamma_i\|^2 +  \sum_j \psi(\lambda_j)\|w_j\|^2
    %     = \sum_i  \varphi(\gamma_i) +  \sum_j \psi(\lambda_j) 
    % \\
    %     &= \Tr\left[\widehat{\varphi}(\Gamma)\right]+ \Tr\left[ \psi'(W) \right].
    % \end{aligned}
    % \end{equation}
    This gives that $\Psi(W)\le \Tr\left[ \psi'(W) \right]$. We prove that $\Psi(W)\ge \Tr\left[ \psi'(W) \right]$. To do so, we choose a specific $\bar{\Gamma}:=\sum_j \Gamma_j \ket{\xi_j} \bra{\xi_j}$ with $\Gamma_j \in \partial \psi(W_j)$. This in particular shows that the inequality in \eqref{eq:computation_transform} is an equality, proving the claimed inequality.
\end{proof}

Incidentally, during the proof above, we explicitly construct a maximizer. 

% \nat{I am not sure understand the logic of this remark bellow very well.} \lore{I would probably take it out, if we agree.}
% \begin{rem} 
% Alternatively, let $\bar{\Gamma} \in \argmax \{ \langle W, \Gamma \rangle - \Tr\left[ \varphi(\Gamma)\right]\}$. By considering perturbations of the form $\Gamma_s := \bar{\Gamma} +s Z$, we have that $W ={\varphi'}(\bar{\Gamma})$. This, in particular, implies that, if $\varphi'$ is invertible and using the composition formula in \eqref{eq:algebra_liftings} we have that $(\varphi')^{-1}= \varphi'^{-1}$ that $\bar{\Gamma}={{\varphi'}^{-1}}(W)$. Then
% \begin{equation}
%     \Psi(W)= \langle W, {{\varphi'}^{-1}(W)} \rangle-\Tr\left[ {\varphi} \circ {\varphi'^{-1}}(W) \right]=\Tr\left[ 
% {{\varphi'}^{-1}(W)} W- {\varphi} \circ {\varphi'^{-1}}(W) \right]= \Tr\left[ 
% {f}(W) \right]
% \end{equation}
% with $f(t):=\varphi'^{-1}(t)t -\varphi \circ {\varphi'}^{-1}(t)$, where in the last equality we used the formulas in \eqref{eq:algebra_liftings}. In particular, notice that, by the same arguments above, it is possible to check that $\psi =f$, thus concluding. \fr
% \end{rem}

The next lemma considers maps induced via functional calculus by convex functions on the real line, and discusses their properties.
\begin{lem}[Gradients of trace functions]
\label{lem:gradient}
    Let $\psi \in C^1(\R)$ a convex function and define the map $\Psi: \rmH(\cH) \to \R$ as $\Psi(A) := \Tr[ \psi (A)]$. Then we have that
    \begin{align}
    \label{eq:gradient}
        \frac{\dd}{\dd t}\Big|_{t=0} \Psi(A+tB)
            = 
        \langle
            \nabla \Psi(A) , B
        \rangle
            =
        \Tr
        \Big[
            \nabla \Psi(A) B 
        \Big]
            \, , \quad \text{where} \quad 
        \nabla \Psi(A) = \psi'(A)
            , \quad 
    \end{align}
    for every $A \in \rmH(\cH)$ and $B \in \rmH(\cH)$. In particular, for every $A \in \rmH(\cH)$ it holds
    \begin{align}
    \label{eq:legendre_equality}
        \Tr[A  \psi'(A)] 
            =
         \Tr[ \psi(A)]
            +
        \Tr[ \psi^* (  \psi'(A))]
            .
    \end{align}
\begin{proof}
The proof of \eqref{eq:gradient} can be found e.g. in  \cite{carlen2010trace}.  The equality in \eqref{eq:legendre_equality} follows directly from  Proposition~\ref{prop:computation_legendre_transform} and the fact that
\begin{align}
        \langle A, \nabla \Psi(A) \rangle
            = 
        \Psi(A) + \Psi^* (\nabla \Psi(A))
            , 
\end{align}
which is a property of real, convex functions and their Legendre transform \cite[Theorem~26.4]{Rockafellar:1970}.
%  \begin{align}
%  \begin{cases}
%      g^*(z_y) + g(y) = \langle z_y, y \rangle 
%           & 
%     \text{where} \quad 
%     y = \nabla g^*(z_y)
%  \\
%      g^*(z_y) + g(y_{z_y}) = \langle z_y , y_{z_y} \rangle 
%         & 
%     \text{where} \quad 
%     z_y = \nabla g(y_{z_y})
% \\
%     y_{z_y} = y  \quad \Rightarrow z_y = \nabla g(y) 
%         &\text{by strict convexity} 
% \end{cases}
%  \end{align}
 \end{proof}
\end{lem}

\subsection*{Acknowledgements}
E.C.\ acknowledges the support of the New Frontiers in Research Fund (NFRFE-2021-00798) and the European Union's Horizon 2020 research and innovation programme (Grant agreement No. 948021). Both A.G.\ and N.M.\ acknowledge the support of the Canada Research Chairs Program, the Natural Sciences and Engineering Research Council of Canada and the New Frontiers in Research Fund (NFRFE-2021-00798).  L.P.\ gratefully acknowledges funding from the Deutsche Forschungsgemeinschaft (DFG, German Research Foundation) under Germany'{s} Excellence Strategy - GZ 2047/1, Projekt-ID 390685813. 
Financial support by the Deutsche Forschungsgemeinschaft (DFG) within the CRC 1060, at University of Bonn project number 211504053, is also gratefully acknowledged. We thank Luca Tamanini for the critical reading of a preliminary version of the manuscript.

\bibliographystyle{abbrv}
\bibliography{biblio.bib}

\end{document}